\documentclass[lettersize,journal]{IEEEtran}
\usepackage{amsmath,amsfonts}
\usepackage{array}
\usepackage{textcomp}
\usepackage{stfloats}
\usepackage{url}
\usepackage{verbatim}
\usepackage{graphicx, subfig}
\usepackage{cite}
\usepackage{amsthm, mathtools}
\usepackage[noend]{algpseudocode}
\usepackage{algorithmicx,algorithm}
\newtheorem{theorem}{Theorem}
\newtheorem{corollary}{Corollary}
\newtheorem{lemma}{Lemma}
\newtheorem{definition}{Definition}

\newenvironment{proofsketch}{\begin{proof}[\textup{\noindent \emph{\textbf{Proof Sketch}}}]}{\end{proof}}

\usepackage{xcolor}
\usepackage{hyperref}
\usepackage{cleveref}

\usepackage{booktabs} % 用于创建三线表

\makeatletter
\newcommand*{\new}{\@ifnextchar\bgroup{\new@}{\color{blue}}}
\newcommand*{\new@}[1]{{\color{blue}{#1}}}

\begin{document}

\title{Regret Minimization in Population Network Games:\\Vanishing Heterogeneity and Convergence to Equilibria}

\author{Die Hu\textsuperscript{*}, \IEEEmembership{}
    Shuyue Hu\textsuperscript{*,\dag}, \IEEEmembership{}
    Chunjiang Mu, \IEEEmembership{}
    Shiqi Fan, \IEEEmembership{}
    Chen Chu, \IEEEmembership{}
    Jinzhuo Liu, \IEEEmembership{}
    Zhen Wang\textsuperscript{\dag}, \IEEEmembership{Fellow,~IEEE}
    \thanks{
        This work was supported by the National Science Fund for Distinguished Young Scholarship of China (No. 62025602), the National Natural Science Foundation of China (Nos. U22B2036, 11971421, 12271471, 62066045), Excellent Youths Project for Basic Research of Yunnan Province (No. 202101AW070015), in part by the Fundamental Research Funds for the Central Universities (Nos. G2024WD0151,  D5000240309), the Tencent Foundation and XPLORER PRIZE.
        (\textsuperscript{\dag} Corresponding\  authors: Shuyue\  Hu, Zhen\  Wang.)

        Die Hu is with the School of Mechanical Engineering, Northwestern Polytechnical University, Xi'an 710072, China and the Department of Computing, Hong Kong Polytechnic University, Hong Kong, China (e-mail: diane.diehu@gmail.com).

        Shuyue Hu is with the Shanghai Artificial Intelligence Laboratory, Shanghai, China (e-mail: hushuyue@pjlab.org.cn).

        Chen Chu is with the School of Statistics and Mathematics, Yunnan University of Finance and Economics, Kunming 650221, China and the School of Artificial Intelligence, Optics and Electronics (iOPEN), Northwestern Polytechnical University, Xi'an 710072, China (e-mail: chuchenynufe@hotmail.com).

        Jinzhuo Liu is with the Department of the School of Software, Yunnan University, Kunming 650091, China and the School of Artificial Intelligence, Optics and Electronics (iOPEN), Northwestern Polytechnical University, Xi'an 710072, China (e-mail: jinzhuo.liu@hotmail.com).

        Shiqi Fan is with the School of Cybersecurity, Northwestern Polytechnical University, Xi'an, China and in Department of Data Science and Artificial Intelligence, the Hong Kong Polytechnic University, Hong Kong, China (e-mail:fsq@mail.nwpu.edu.cn).

        Zhen Wang, and Chunjiang Mu are with the School of Cybersecurity and the School of Artificial Intelligence, Optics and Electronics (iOPEN), Northwestern Polytechnical University, Xi'an 710072, China (e-mail: zhenwang0@gmail.com; mcj\_nwpu@163.com).
    }
}

\maketitle

\begin{abstract}
    Understanding and predicting the behavior of large-scale multi-agents in games remains a fundamental challenge in multi-agent systems. This paper examines the role of heterogeneity in equilibrium formation by analyzing how smooth regret-matching drives a large number of heterogeneous agents with diverse initial policies toward unified behavior. By modeling the system state as a probability distribution of regrets and analyzing its evolution through the continuity equation, we uncover a key phenomenon in diverse multi-agent settings: the variance of the regret distribution diminishes over time, leading to the disappearance of heterogeneity and the emergence of consensus among agents. This universal result enables us to prove convergence to quantal response equilibria in both competitive and cooperative multi-agent settings. Our work advances the theoretical understanding of multi-agent learning and offers a novel perspective on equilibrium selection in diverse game-theoretic scenarios.

\end{abstract}


\IEEEpeerreviewmaketitle

\section{Introduction}
\IEEEPARstart{U}{nderstanding} the dynamics of learning behavior in games has long been a fundamental problem studied in multi-agent learning~\cite{zhang2018data,zhang2020rnn,li2022online}, contributing to numerous applications \cite{heusel2017gans,tembine2019deep,jiang20233d,jiang2023dynamic}.
Regret minimization stands out as fundamental to the development and learning of algorithms, leveraging past decisions to offer a practically efficient method for approximating game-theoretic equilibria.
Informally, in the context of opponents' strategies, regret minimization iteratively optimizes a player's strategy by minimizing the distance between the learner's actions and either the best fixed action in hindsight (external regret \cite{hannan1957approximation}) or a fixed strategy (swap regret \cite{blum2007external}).

It is well known that Regret minimization and its many variants can guarantee time-average convergence to certain classes of equilibria in games~\cite{brown2018superhuman,brown2019superhuman}. For instance, in general-sum games, the time-averaged strategy profile gradually converges to a coarse correlated equilibrium, while in wo-player zero-sum games, it converges to a Nash equilibrium. Hart and Mas-Colell~\cite{2000hartSimpleAdaptiveProcedure} proposed the \emph{regret matching} algorithm whereby agents track regrets to guide deviation from current strategies in proportion to accumulated non-negative regrets, making decision-making more adaptive and look forward. Over the years, due to their solid theoretical guarantees and strong empirical performance, regret matching and its variant, particularly in the framework of counterfactual regret minimization, have become an essential ingredient for the practical state of the art in multiple breakthroughs such as Computer Poker Research \cite{brown2019superhuman, hart2013simple,zhang2021multi,zhang2023dynamics}.

However, such results do not necessarily imply that the \emph{real-time} behavior of system always converges fully to these equilibrium points. Even in the simplest class of games, such as two-player two-action zero-sum games, the real-time dynamics of the system can exhibit chaotic, non-convergent behavior \cite{bailey2019fast,zhu2020online,chen2022online,piliouras2022evolutionary}. This is especially true in more complex multi-agent systems \cite{chotibut2020route,bielawski2021follow,czechowski2023non}, where convergence is even harder to guarantee. Furthermore, while the regret matching algorithm and its variations have been proven to converge in potential games, this convergence does not always extend to converge towards an average-$\epsilon$-Nash equilibrium instead \cite{blum2010routing}. In fact, such large-scale multi-agent systems are particularly appealing and have attracted increasingly more interest \cite{blum2006routing,lam2016learning,chotibut2020route,luo2021multiagent,huHeterogeneousBeliefsMultiPopulation2023}. In real life, many multi-agent systems -- such as traffic networks, stock markets, and smart grids -- naturally involve a great number of agents, and the property of heterogeneity widely exists~\cite{li2017global,hui2019game}.
This raises a series of fundamental questions:
\emph{will the dynamic of a system of heterogeneous regret-minimizing agents equilibrate? If so, what class of equilibria does the real-time system behavior evolve in, and what can we say about the heterogeneity?}

Motivated by these problems, in this paper, we study a variant of regret matching (called \emph{smooth regret matching}) and consider a common large-scale multi-agent setting.
Specifically, this setting contains \emph{a continuum of heterogeneous agents} with diverse policies initially; over time, agents adapt their policies through smooth regret matching. In recent work that we build upon, Wang et al.~\cite{2022wangModellingDynamicsRegret} established a partial differential equation, akin to the \emph{continuity equation} commonly encountered in the studies of physical systems, to model the dynamics of a continuum of heterogeneous smooth regret-matching agents. Whereas their finding connects two seemingly unrelated literature of regret minimization and continuity equation, the above questions still lack conclusive answers.

\emph{\textbf{Our model.}}
Our paper answers the above questions about smooth regret matching in a general class of population network games (PNG). In analogy to the conventional network game \cite{kearns2001graphical}, a PNG is defined over a graph where each vertex represents a population (or continuum) of agents, and each edge defines a series of 2-player subgames between individuals from two adjacent populations.
PNGs are game-theoretic models for real-world systems where the graph structure is evident (such as multi-army combats, multi-national international trade, or traffic networks) \cite{czechowski2022poincare} and encompass all the 2-population games commonly studied in evolutionary game theory \cite{sandholm2010population,li2016distributed,liu2018modeling}.
At every time step, each individual of a population plays a policy in all its 2-player subgames against its opponents from the adjacent populations.
Agents apply smooth regret matching to track regrets for past plays and to select future plays based on regrets using the softmax function.

\emph{\textbf{Our Approach \& Contributions.}}
We define the system state as a probability distribution over the space of regrets.
Such a definition of the system state follows \cite{2022wangModellingDynamicsRegret} and naturally accounts for heterogeneous agents (with different regrets and policies) under our concerned setting.
When agents learn in games and update their regrets, this probability distribution changes correspondingly.
We track this change and show in \Cref{theorem:regretdynamics} that in the continuous-time limit, the dynamic of that distribution follows the continuity equation.
This theorem generalizes that of previous work \cite{2022wangModellingDynamicsRegret} to PNGs.
This indicates that the link between the literature on regret minimization and the continuity equation, which remains underexplored, could have significant and far-reaching implications.

However, the continuity equation described by PDEs is challenging to solve directly. Going beyond previous work, we adopt the moment closure technique~\cite{goodman1953population,whittle1957use}, which approximates PDEs by expressing higher-order terms through lower-order ones, making them tractable. This approach is particularly useful for analyzing large systems, such as populations or particle ensembles. The \textit{mean} regret (first moment) represents the average system state, while the \textit{variance} of regrets (second moment) quantifies system heterogeneity.
Theorem \ref{theorem:heto_dERdt} reveals that a population's mean regret is intricately linked to variance, indicating that system heterogeneity influences the evolution of the average state. More importantly, leveraging variance properties, we prove in \Cref{theorem:VarDecay} that regret variance decreases over time and asymptotically vanishes. Notably, our proof does not depend on the 2-player subgame payoff structure or specific action selection functions, provided they are continuously differentiable with respect to regrets.
Thus, vanishing heterogeneity appears to be a general phenomenon in PNGs, despite contradicting intuitive expectations. This insight deepens our understanding of regret dynamics, shedding light on the natural homogenization of large populations.

As another key contribution, we provide a convergence guarantee for learning in PNGs, showing that heterogeneous populations converge to a specific type of quantal response equilibrium (QRE), a refinement of Nash equilibrium that accounts for bounded rationality. Our proof techniques, of independent interest, bridge smooth regret matching with the well-studied smooth Q-learning algorithm.
Using the asymptotically autonomous theorem \cite{1992thiemeConvergenceResultsPoincareBendixson}, we prove that policy dynamics become asymptotically autonomous, with the limit equation resembling smooth Q-learning dynamics \cite{2003tuylsSelectionmutationModelQlearning} (shown in \Cref{lemma:dxdt}). This implies that as heterogeneity vanishes, policy dynamics align with smooth Q-learning. Building on this connection, we establish our convergence result in \Cref{theorem:convergence}, proving that under smooth regret matching, agents' policies converge to a unique QRE in weighted zero-sum PNGs and to a compact connected set of QREs in weighted potential PNGs.
This result highlights convergence behaviors in both competitive and cooperative settings, covering most two-population cases. It applies to a wide range of strategic interactions, from adversarial scenarios like cybersecurity and financial markets to cooperative ones such as traffic routing and energy management, making it broadly relevant across disciplines.

Last but not least, we validate our theoretical findings through agent-based simulations across six game types and three network structures. Our results show that agents' policies become increasingly concentrated over time, ultimately converging to a common policy, even from highly diverse initial distributions, confirming vanishing heterogeneity. Additionally, by analyzing agents' mean policy over time, we verify its convergence to a quantal response equilibrium, aligning with our theoretical predictions. In cases with multiple QREs, our experiments suggest that system heterogeneity may influence equilibrium selection. These findings reinforce our theoretical insights and highlight the broader implications of heterogeneity in strategic interactions.

To summarize, our key contributions are as follows.
\begin{itemize}
    \item We model the real-time regret dynamic in population network games with the continuity equation, thereby broadening the formal connection between two seemingly unrelated literature: regret minimization and continuity equation.
    \item We uncover a widespread phenomenon of vanishing heterogeneity among a large number of regret-minimizing agents, which has not been reported previously.
    \item We provide a convergence guarantee for regret minimization in weighted zero-sum and weighted potential population network games.
    \item We empirically corroborate our theory with extensive numerical simulations and agent-based simulations in six different types of games and three network structures.
\end{itemize}

\section{Preliminaries}\label{sec:preliminaries}

This section defines the population network game and the smooth regret matching.

\subsection{Population Network Games}
Population network games (PNGs, also called graphical polymatrix games) have attracted much recent interest \cite{kearns2001graphical}; they encompass all the 2-player games, and easily identifiable subclasses of real-world systems where the graph structure is evident \cite{czechowski2022poincare}. These games capture intricate dynamics across populations, as detailed in~\cite{szabo2007evolutionary}. For instance, in combat, each army, as a node, engages in competition and cooperation with others and is controlled by independent learning agents based on local observations. Moreover, in ecology, nodes represent species or populations, and within each, individual organisms act independently, contributing to strategic diversity and mutual influence through interactions like predation and symbiosis.

A population network game
$\Gamma=((V,E),(S_i)_{\forall i\in V},(\mathbf{A}_{ij})_{(i,j)\in E})$
is defined over a graph $(V, E)$, where $V=\{1,\ldots,n\}$ is the set of vertices each represents a population (or continuum) of agents, and $E$ is the set of edges each defines a series of two-player subgames between two populations.
For each population $i\in V$, agents of this population have a finite set $S_i$ of actions with generic elements $a\in S_i$.
We denote the policy of a generic agent $k$ of population $i\in V$ by a vector $x_{i_k}\in \Delta_i$, where $\Delta_i=\{x_{i_k}\in\mathbb{R}^{\vert S_i\vert}\vert\sum_{a\in S_i}x_{i_k}(a)=1,x_{i_k}(a)\geq0,\forall a\in S_i\}$ is a probability simplex over the action set $S_i$.
For each edge $(i,j) \in E$, every agent in population $i$ is randomly paired up with another agent in population $j$ to play a two-player subgame.
Subgames don't have to be the same across different pairs; however, it's crucial that the dimensions of action spaces are consistent across them.
Given the two-player subgames between population $i$ and $j$, we denote the payoff matrix of populations $i$ and population $j$ by $\mathbf{A}_{ij}$ and $\mathbf{A}_{ji}$, respectively.
Note that at any given time step, each agent chooses a policy and plays that policy in all its two-player subgames.
Let $v_i$ be the set of agents in population $i$ and $V_i=\{j\in V:(i,j)\in E\}$ be the set of population $i$'s adjacent populations.
Given a mixed strategy profile $(\mathbf{x}_{i_k}, \{\mathbf{x}_{j}\}_{j\in V_i})$, the expected payoff for agent $k$ in population $i$ in all its two-player subgames is
$$
    u_{i}(\mathbf{x}_{i_k}, \{\mathbf{x}_{j}\}_{j\in V_i})=\frac{1}{\vert V_i\vert}\sum_{j\in V_i}\frac{1}{\vert v_j\vert}\sum_{k' \in v_j}\mathbf{x}_{i_k}^\top \mathbf{A}_{ij}\mathbf{x}_{j_{k'}},
$$
where $k'$ denotes an arbitrary agent in population $j$.
According to the definition, the payoff $u_i(\mathbf{x}_{i_k}, \{\mathbf{x}_{j}\}_{j\in V_i})$ encapsulates the game payoff from all interactions with agent $k$, assuming an infinite number of agents in each population and, consequently, an infinite number of interactions within each population. In view of the linearity of expectation, meaning that $\mathbb{E}[\mathbf{x}_{j}^{t}]=\frac{1}{v_j}\sum_{k'\in j}\mathbf{x}_{j_{k'}}\approx\bar{\mathbf{x}}_{j}^t$.
Then we rewrite the expected payoff function as
\begin{equation}
    u_{i}(\mathbf{x}_{i_k}, \{\mathbf{x}_{j}\}_{j\in V_i})=\frac{1}{\vert V_i\vert}\sum_{j\in V_i}\mathbf{x}_{i_k}^\top\bar{\mathbf{A}}_{i_kj}\bar{\mathbf{x}}_{j},
\end{equation}
where $\bar{\mathbf{A}}_{i_kj}$ is the weighted payoff matrix for agent $k$ interacting with $j$ populations.

Put differently, $\mathbf{r}_{i_k}^t$ can be considered as the mean instantaneous regret $\bar{\mathbf{r}}_{i_k}^t$. In the following, we remark $\mathbf{r}_{i_k}^t$ as $\bar{\mathbf{r}}_{i_k}^t$.
In view of the linearity of expectation, $\bar{\mathbf{r}}_{i_k}^t$ also can be considered as an unbiased estimator, meaning that $\mathbb{E}[\mathbf{r}_{i_k}^{t}]\approx\bar{\mathbf{r}}_{i_k}^t$.
Moreover, from Hu et al.~\cite{2019huModellingDynamicsMultiagent} and Wang et al.~\cite{2022wangModellingDynamicsRegret}, the mean policy $\bar{\mathbf{x}}_{j\in V_i}$ is actually an unbiased estimator of the population $j$'s expected policy, and the variance of that estimator equals 0.

We remark that our PNG generalizes the traditional network game, considering that a vertex can represent either a heterogeneous or a homogeneous agent population.
If each vertex in our PNG represents a homogeneous agent population, then our PNG is effectively the same as the traditional network game.

\subsection{Smooth Regret Matching}
By convention, we define instantaneous regret of not having chosen an action as the difference between the instantaneous payoff of that action and the expected payoff of the strategy actually in use.
Consider an arbitrary agent $k$ in a population $i\in V$. At each time step $t$, agent $k$ takes the strategy $\mathbf{x}_{i_k}$ to play with all opponents from adjacent populations. Let $\mathbf{r}_{i_k}=[\mathbf{r}_{i_k}^t(a)\vert\forall a\in S_i]\in\mathbb{R}^{\vert S_i\vert}$ be the instantaneous regret of not having chosen that action, given by
\begin{equation}\label{eq:Erika}
    \begin{aligned}
        \mathbf{r}_{i_k}^t(a) & =u_{i_k}(a,\bar{\mathbf{x}}_{-i})-u_{i_k}(\mathbf{x}_{i_k}^t,\bar{\mathbf{x}}_{-i})                                     \\
                              & =(\mathbf{e}_{i_ka}^\top-\mathbf{x}_{i_k}^t)\frac{1}{\vert V_i\vert}\sum_{j\in V_i}\mathbf{A}_{ij}\bar{\mathbf{x}}_j^t,
    \end{aligned}
\end{equation}
where $\mathbf{e}_{i_ka}$ is a unit vector where the $a$-th element is 1. For notation convenience, sometimes we write $u_{i_ka}^t(\mathbf{x}_{-i})\coloneqq u_{i_k}(a,\{\mathbf{x}_j^t\}_{(i,j)\in E})$.
Agent $k$'s cumulative regret of not having chosen each action $a\in S$ over history is defined as
\begin{equation}\label{eq:Ria}
    \mathbf{R}_{i_k}^{t+1}(a)=\frac{1}{t}\sum_{\tau=1}^t\mathbf{r}_{i_k}^\tau(a).
\end{equation}
We let $\mathbf{R}_{i_k}$ be the vector of cumulative regrets. For simplicity, we call $\mathbf{R}_{i_k}$ to be agent $k$'s regret in the following.

In \cite{2000hartSimpleAdaptiveProcedure},
regrets are considered to be non-negative, and agents' probabilities of choosing an action are proportional to the regret for not having chosen that action. Wang et al.~\cite{2022wangModellingDynamicsRegret} generalized regret matching to allow negative regrets and adopted the softmax function for action selection:
\begin{equation}\label{eq:xia}
    \mathbf{x}_{i_k}^{t+1}(a)=\frac{\exp[\lambda \mathbf{R}_{i_k}^{t}(a)]}{\sum_{b\in S_i}\exp[\lambda \mathbf{R}_{i_k}^{t}(b)]},
\end{equation}
where $\lambda\in [0,+\infty)$ denotes the temperature. When $\lambda=0$, the agent always acts completely random and chooses the strategy uniformly. Whereas when $\lambda\to +\infty$, the agent greedily chooses the action with the highest regret value.

\begin{algorithm}
    \caption{Smooth Regret Matching in Population Network Games}
    \label{alg:smooth_regret_matching}
    \begin{algorithmic}[1]
        \State Initialize the regret matrix $R$ for each agent with zeros.
        \State Initialize the policy matrix $policy$ with an arbitrary distribution.
        \For{each round $t$}
        \For{each population $i$}
        \For{each agent $k$}
        \State Calculate the instantaneous regret by \Cref{eq:Erika}.
        \State Update the regret by \Cref{eq:Ria}.
        \EndFor
        \EndFor
        \ForAll{agents $k$}
        \State Update the policy by \Cref{eq:xia}.
        \EndFor
        \EndFor
    \end{algorithmic}
\end{algorithm}

%---------------

\begin{table}[!t]
    \caption{Important Notations}
    \scalebox{1.15}{
        \begin{tabular}{ll}
            \hline
            \textbf{Symbol}        & \textbf{Description}                                 \\
            \hline
            $\Gamma$               & The population network game                          \\
            $i$                    & Population $i$                                       \\
            $V_i$                  & The neighbour population set of population $i$       \\
            $S_i$                  & The action set of population $i$                     \\
            $E=\{e_{ij}\}$         & The relationship between two populations $i$ and $j$ \\
            $\Delta_i$             & Strategy spaces for agents of population $i$         \\
            $\mathbf{r}_{i_k}$     & Instantaneous regret of agent $k$                    \\
            $\mathbf{R}_{i_k}$     & Regret of agent $k$                                  \\
            $\mathbf{x}_{i_k}$     & Strategy profile of agent $k$                        \\
            $p(\mathbf{R}_{i}, t)$ & Probability density of regrets                       \\
            $\lambda$              & Temperature parameter                                \\
            \hline
        \end{tabular}
    }
\end{table}

\section{Modelling Regret Dynamics in Population Network Games}
\label{sec:dynamics}
In this section, we model the regret dynamic in population network games. We shall first describe the regret dynamic with a partial differential equation (PDE) using the master equation approach \cite{2022wangModellingDynamicsRegret}.
Then, going beyond \cite{2022wangModellingDynamicsRegret}, we shall adopt the moment closure technique and establish two ordinary differential equations (ODEs) that govern the time evolution of the population trend and the population heterogeneity.

We consider the system state to be a probability distribution over the space of regrets $\Omega=\prod_{i\in V}\Omega_i$.
We denote the probability density of regrets $R_i$ at time $t$ by $p(\mathbf{R}_i,t)$. Intuitively, the density $p(\mathbf{R}_i,t)$ indicates the fraction of agents having the regret $\mathbf{R}_i$ in population $i$. In the following theorem, we show that the time evolution of $p(\mathbf{R}_i,t)$ is governed by a PDE.

\begin{theorem}[Population Dynamics]
    \label{theorem:regretdynamics}
    For each population $i\in V$, the dynamic of the probability density function $p(\mathbf{R}_i,t)$ is governed by
    For each population $i\in V$, the dynamic of the probability density function $p(\mathbf{R}_i,t)$ is governed by
    \begin{equation}\label{eq:PDF_R1}
        \frac{\partial p(\mathbf{R}_i,t)}{\partial t}=-\nabla \cdot[p(\mathbf{R}_i,t)\frac{\bar{\mathbf{r}}_i-\mathbf{R}_i}{t}],\forall i\in V
    \end{equation}
    where $\nabla \cdot$ is the divergence operator, and $\bar{\mathbf{r}}_i$ is the mean immediate regret of population $i$ such that $\forall a\in S_i$,
    \begin{equation}\label{eq:Eria_int}
        \begin{split}
            \bar{r}_{ia}= & \int\int_{\prod_{j\in V_i}}[\mathbf{e}_{ia}-\mathbf{x}_i]^\top\frac{1}{\vert V_i\vert}\sum_{j\in V_i}\mathbf{A}_{ij}\bar{\mathbf{x}}_j \\
                          & (p(\mathbf{R}_i, t)\prod_{j\in V_i}p(\mathbf{R}_j,t))(d\mathbf{R}_i\prod_{j\in V_i}d\mathbf{R}_j),
        \end{split}
    \end{equation}
    where $\bar{\mathbf{x}}_j=\int_{\prod_{b\in S_j}}\frac{\exp(\lambda\mathbf{R}_{jb})}{\sum_{b'\in S_j}\exp(\lambda\mathbf{R}_{jb'})}(\prod_{b\in S_j}d\mathbf{R}_{jb})$.
\end{theorem}

See Appendix A for a detailed certification. If we let $\rho_i= p(\mathbf{R}_i,t)$ and $\vec{v}_i=\left(\bar{\mathbf{r}}_i-\bar{\mathbf{R}}_i\right)/t$, then we can recover \Cref{eq:PDF_R1} as a \emph{continuity equation}
$\partial_t \rho_i + \nabla\cdot(\rho_i\mathbf{v}_{i})= 0$.
The continuity equation is a widely used equation in physics and mathematics that describes systems where a conserved physical quantity (such as mass, energy, momentum, or charge) remains constant. Any changes in the quantity occur solely due to inflow or outflow through a boundary, ensuring that the quantity neither vanishes nor appears spontaneously.
Thus, this theorem implies that the continuity equation, which typically describes the transport of conserved quantity in a physical system, can also provide an accurate description of the learning dynamics under our concerned setting.
We remark that \Cref{theorem:regretdynamics} generalizes that of Wang et al. \cite{2022wangModellingDynamicsRegret} to allow for multiple populations and graph structure; in their work, they consider only one population in which agents are randomly paired up to play 2-player symmetric games.

However, our continuity \Cref{eq:PDF_R1} is unfortunately analytically unsolvable. In fact, only very few cases, e.g. one-dimensional linear continuity equations, allow for analytic solutions. As a remedy, we employ moment closure techniques to address the effects of higher-order moments by approximating them through the coefficients of the Taylor expansion, as detailed in the Appendix. By truncating the expansion of a certain order, we are able to both capture the underlying system behavior and alleviate the computational difficulties associated with the direct inclusion of higher-order moments. Here, we turn our attention to analyzing the moments (e.g. the mean and the variance) of the probability distribution $p(\mathbf{R}_i, t)$ based on the continuity equation.
In the following theorem, we show that for each population $i\in V$, the time evolution of the mean regret $\bar{\mathbf{R}}_i$ (the first moment) can be described by an analytically more amenable ODE.

\begin{theorem}[Mean Regret Dynamics]\label{theorem:heto_dERdt}
    For each population $i\in V$, the continuous-time dynamic of the mean regret $\bar{\mathbf{R}}_i$ can be described by
    \begin{equation}\label{eq:heto_dERdt}
        \begin{split}
            \frac{d\bar{\mathbf{R}}_i}{dt}= & \frac{\bar{\mathbf{r}}_{i}-\bar{\mathbf{R}}_{i}}{t}\approx\frac{1}{t}[f(\{\bar{\mathbf{R}}_j\}_{j\in V_i})-\bar{\mathbf{R}}_{i}] \\
                                            & +\frac{1}{2t}[f''(\{\bar{\mathbf{R}}_j\}_{j\in V_i})\operatorname{Var}(\{\mathbf{R}_j\}_{j\in V_i})],
        \end{split}
    \end{equation}
    where $\operatorname{Var}(\{\mathbf{R}_j\}_{j\in V_i})$ is the variance of regrets in population $j\in V_i$, and $f(\{\bar{\mathbf{R}}_j\}_{j\in V_i})=[\mathbf{e}_{ia}-\mathbf{x}_i]^\top\frac{1}{\vert V_i\vert}\sum_{j\in V_i}\mathbf{A}_{ij}\bar{\mathbf{x}}_j$ is the instantaneous regret of population (\Cref{eq:Eria_int} implies the the integral of each dimension).
\end{theorem}

Intuitively, for each population, the mean regret dynamics \eqref{eq:heto_dERdt} represents its trend of change in regrets.
Thus, from \Cref{eq:heto_dERdt}, it is clear that such a trend of change in regrets is affected not only by the mean regret itself but also by the variance of regrets.
This finding is non-trivial, as the mean dynamics is usually self-organizing (without the dependence on other moments) in most previous studies of population dynamics \cite{sandholm2010population}.

It is then natural to ask: how does the variance of regrets (the second moment) change over time?
We show in the following theorem that the dynamics of the variance is also governed by an ODE.
More interestingly, based on the ODE, we uncover that the variance of regrets asymptotically tends to zero.
Put differently, as time goes to infinity, all the agents in a population eventually converge to the same regret and, consequently, the same policy.

\begin{theorem}[Decay of the Variance]\label{theorem:VarDecay}
    For each population $i\in V$,
    the dynamics of the variance of regrets $\operatorname{Var}(\mathbf{R}_i)$  is governed by
    \begin{equation}
        \frac{d\operatorname{Var}(\mathbf{R}_i)}{dt}=-\frac{2\operatorname{Var}(\mathbf{R}_i)}{t},
    \end{equation}
    such that at given time $t$, $\operatorname{Var}(\mathbf{R}_i)=\frac{\sigma^2(\mathbf{R}_i)}{t^2}$, where $\sigma^2(\mathbf{R}_i)$ is the initial variance of $\mathbf{R}_i$. Thus, over time, the variance $\operatorname{Var}(\mathbf{R}_i)$ will decay to zero.
\end{theorem}

Note that the proof of the above results is independent of the action selection function.
Hence, \Cref{theorem:VarDecay} implies that no matter what action selection function agents adopt, all the agents in a population will eventually converge to the same policy, regardless of the initial distribution of regrets, the specific two-player sub-games played by agents, or the number of populations and strategies.
We remark that in weakly related work, vanishing heterogeneity is also proved in the studies of smooth fictitious play \cite{fudenberg2011heterogeneous,huHeterogeneousBeliefsMultiPopulation2023}; however, heterogeneity generally remains for the use of Q-learning and Cross learning \cite{aamashu2022,leungModellingDynamicsMultiAgent2022}.

The core insight lies in population-level dynamics: agents share identical payoffs and observations, interacting with the same neighbors. Despite initial heterogeneity, all regrets update based on the same feedback—opponents' average performance—causing regrets to evolve in parallel, align with the population average, and pull outliers toward the mean, leading to vanishing heterogeneity.
It is noteworthy that our theoretical results in \Cref{theorem:regretdynamics} to \Cref{theorem:VarDecay} are universal and hold under the standard assumption of differentiable regret learning dynamics—a common and well-behaved condition for probability distributions. While our current analysis is centered on normal-form games, our framework is readily applicable to other game types, such as extensive-form games, provided they involve two-player subgame scenarios. This broad applicability stems from the fact that every extensive-form game can be uniquely mapped to its corresponding normal-form representation through equivalent payoff matrices ~\cite{cressman2003evolutionary}.

\begin{table}[hb]
    % Set the color of the entire table to blue
    \centering
    \caption{Payoff bimatrices of six different games.}
    \label{tab:payoff}
    \vspace{-0.2cm}
    \subfloat[Presidential election]{
        \begin{minipage}[t]{0.5\linewidth}
            \centering
            \begin{tabular}{lcc}
                \toprule
                  & M      & T      \\
                \midrule
                E & (3,-3) & (-1,1) \\
                S & (-2,2) & (1,-1) \\
                \bottomrule
            \end{tabular}
        \end{minipage}
    }
    \hspace{-1cm}
    \subfloat[Rock-paper-scissors]{
        \begin{minipage}[t]{0.5\linewidth}
            \centering
            \resizebox{0.75\linewidth}{!}{  % Scale the sub-table to 75% of its original width
                \begin{tabular}{lccc}
                    \toprule
                      & R      & P      & S      \\
                    \midrule
                    R & (0,0)  & (1,-1) & (-1,1) \\
                    P & (-1,1) & (0,0)  & (1,-1) \\
                    S & (1,-1) & (-1,1) & (0,0)  \\
                    \bottomrule
                \end{tabular}
            }
        \end{minipage}
    }
    \vspace{-0.2cm}
    \subfloat[Asymmetric matching pennies]{
        \begin{minipage}[t]{0.5\linewidth}
            \centering
            \begin{tabular}{lcc}
                \toprule
                  & H      & T      \\
                \midrule
                H & (2,-4) & (-2,4) \\
                T & (0,0)  & (2,-4) \\
                \bottomrule
            \end{tabular}
        \end{minipage}
    }
    \hspace{-1cm}
    \subfloat[Prisoner's dilemma]{
        \begin{minipage}[t]{0.5\linewidth}
            \centering
            \begin{tabular}{lcc}
                \toprule
                  & C     & D     \\
                \midrule
                C & (6,6) & (2,8) \\
                D & (8,2) & (2,2) \\
                \bottomrule
            \end{tabular}
        \end{minipage}
    }
    \vspace{-0.1cm}
    \subfloat[Stag hunt]{
        \begin{minipage}[t]{0.5\linewidth}
            \centering
            \begin{tabular}{lcc}
                \toprule
                  & S       & H     \\
                \midrule
                S & (10,10) & (1,8) \\
                H & (8,1)   & (5,5) \\
                \bottomrule
            \end{tabular}
        \end{minipage}
    }
    \hspace{-1cm}
    \subfloat[Battle of the sexes]{
        \begin{minipage}[t]{0.5\linewidth}
            \centering
            \begin{tabular}{lcc}
                \toprule
                  & F      & B      \\
                \midrule
                F & (10,5) & (0,0)  \\
                B & (0,0)  & (5,10) \\
                \bottomrule
            \end{tabular}
        \end{minipage}
    }
\end{table}
\begin{figure*}[htbp]
    \centering
    \includegraphics[width=\textwidth]{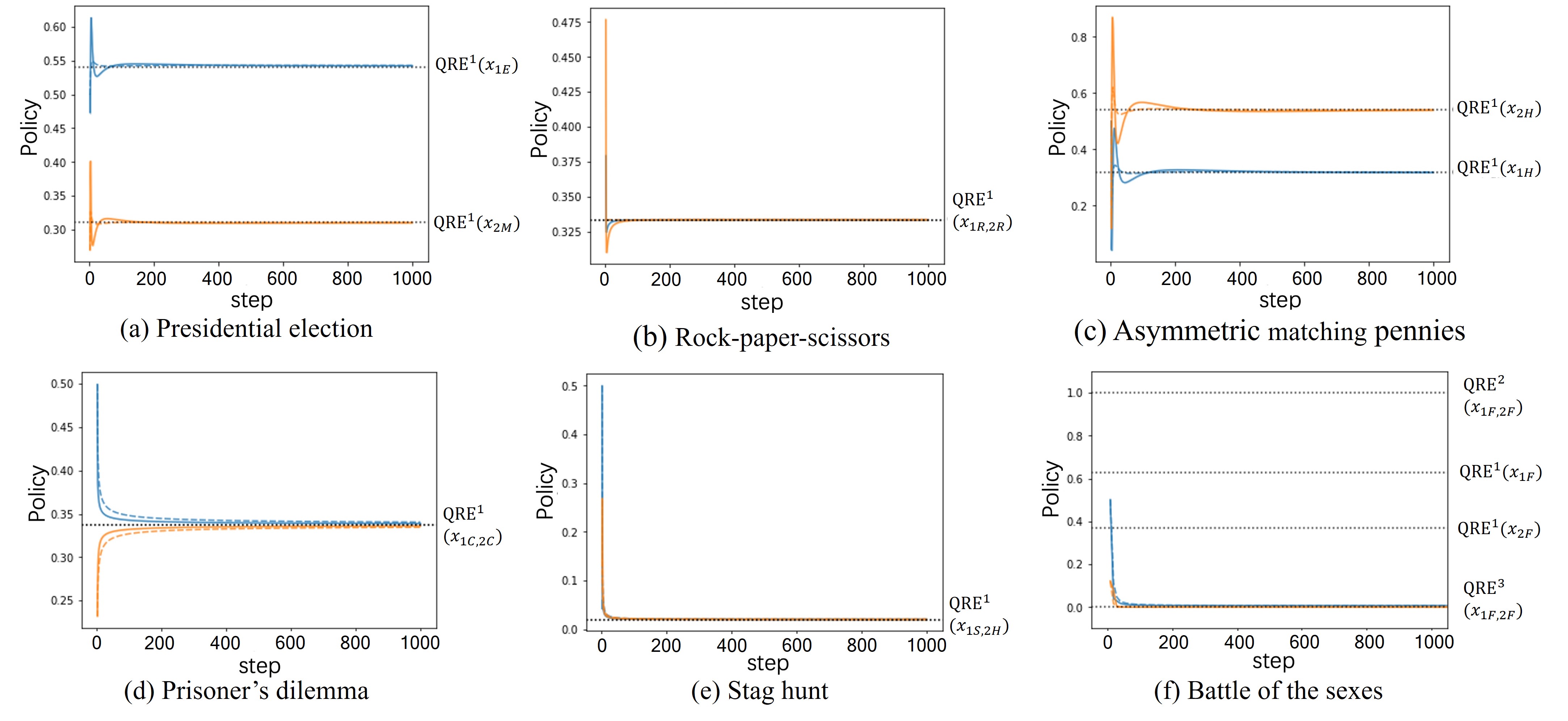}
    \caption{Evolution of the expected policy (dashed lines) derived from our model and that of the mean policy (solid lines) in agent-based simulations of 2 population network games. All games eventually converge to the theoretical QRE (black dotted line).}
    \label{fig:sim_2population}
\end{figure*}

\begin{figure*}[h]
    \centering
    \subfloat[PD, Norm($\sigma=0.05$)\label{fig:ws_pd_var0.05}]{
        \includegraphics[width=0.23\textwidth]{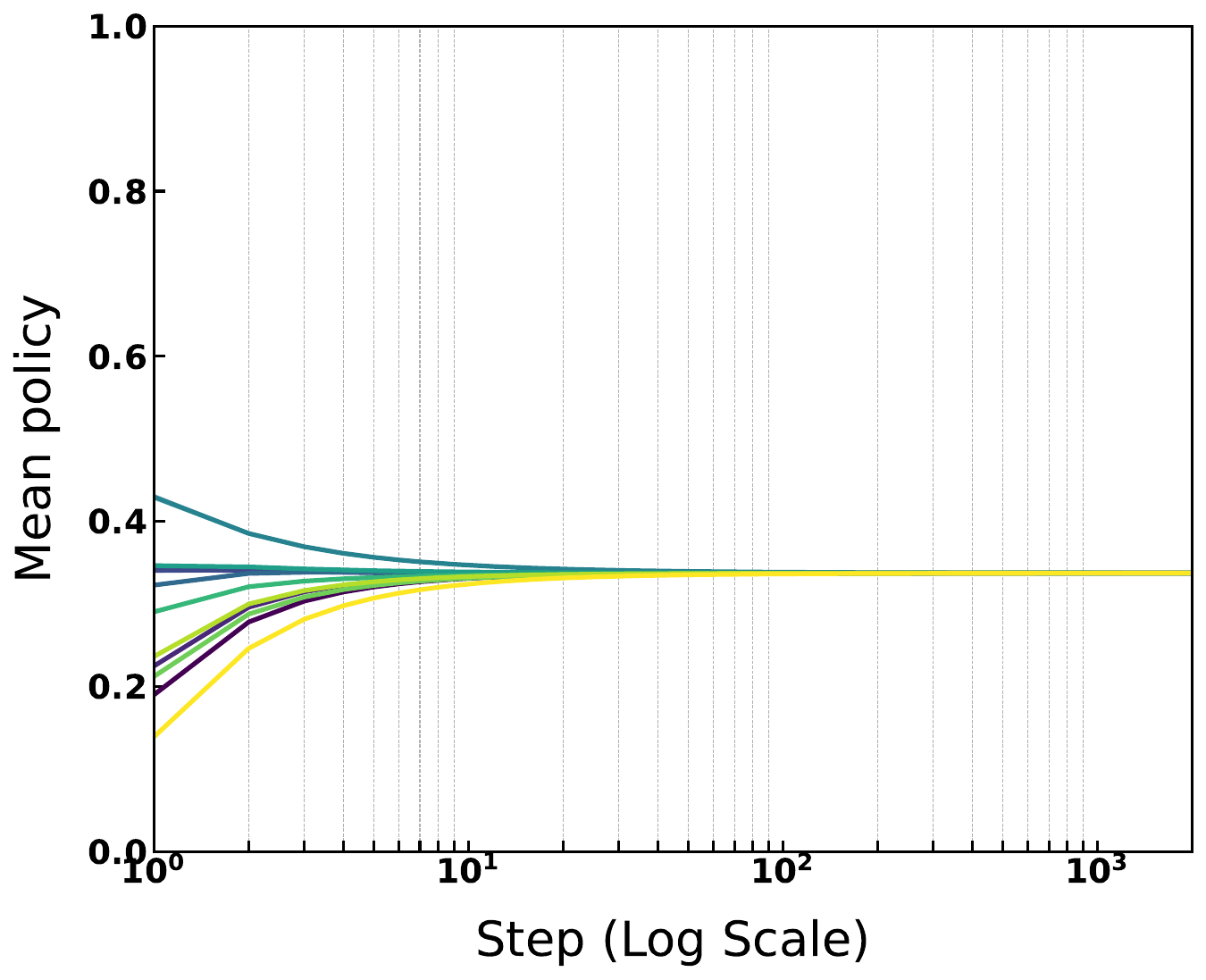}
    }
    \subfloat[PD, Norm($\sigma=0.1$)\label{fig:ws_pd_var0.1}]{
        \includegraphics[width=0.23\textwidth]{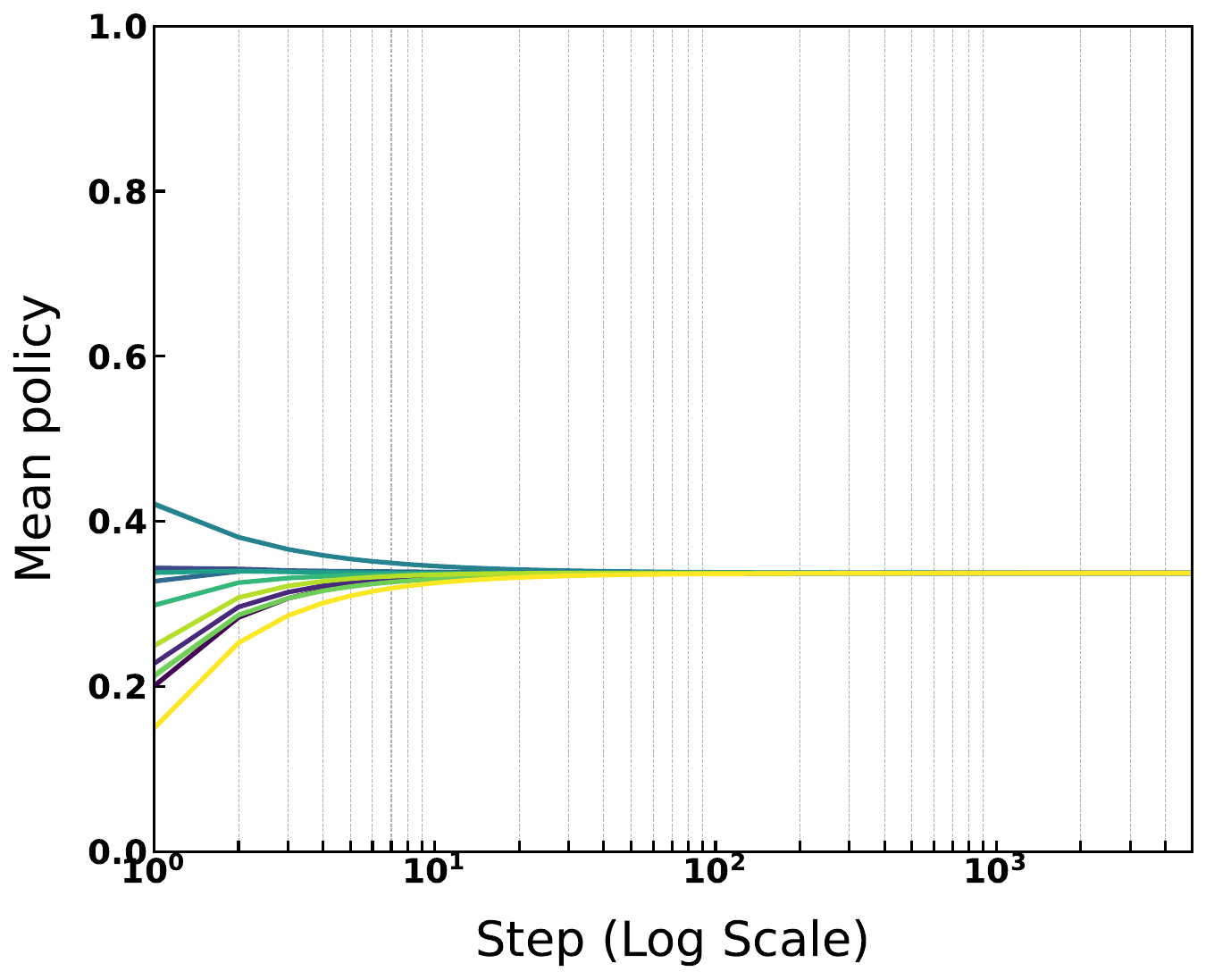}
    }
    \subfloat[RPS, Norm($\sigma=0.05$)\label{fig:ws_rps_var0.05}]{
        \includegraphics[width=0.23\textwidth]{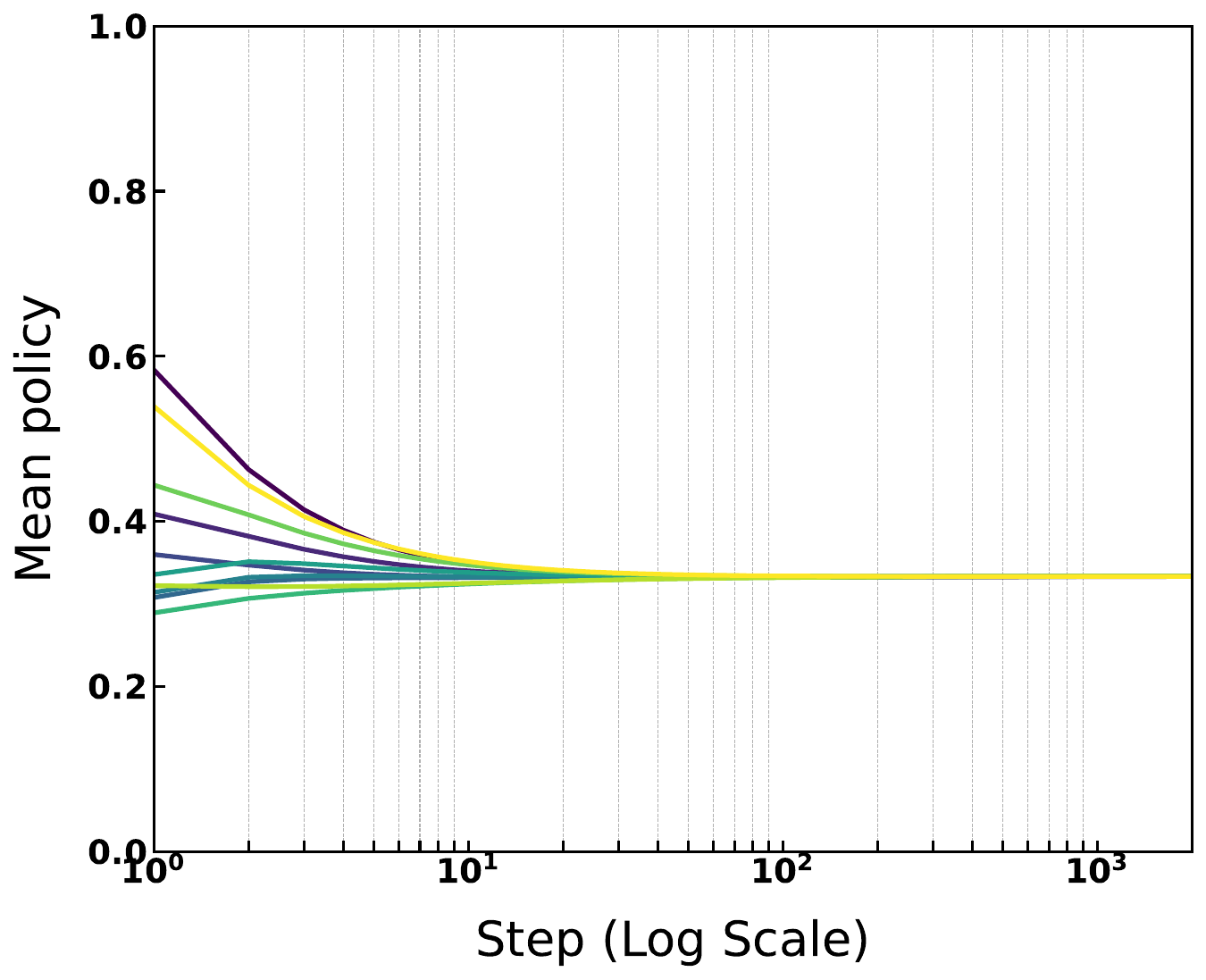}
    }
    \subfloat[RPS, Norm($\sigma=0.1$)\label{fig:ws_rps_var0.1}]{
        \includegraphics[width=0.23\textwidth]{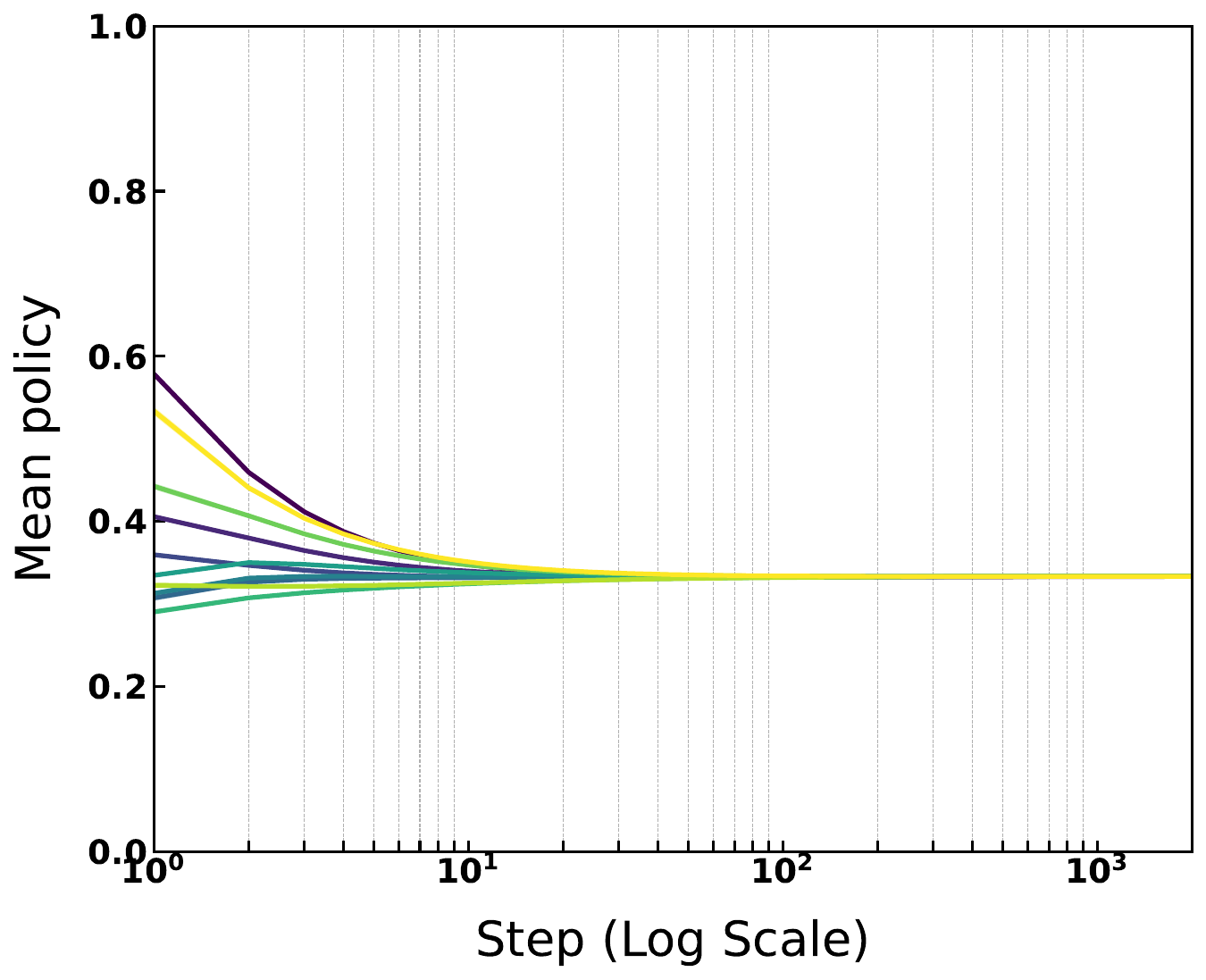}
    }
    \vfill
    \subfloat[PD, Norm($\sigma=0.05$)\label{fig:ws_pd_var0.05_vr}]{
        \includegraphics[width=0.23\textwidth]{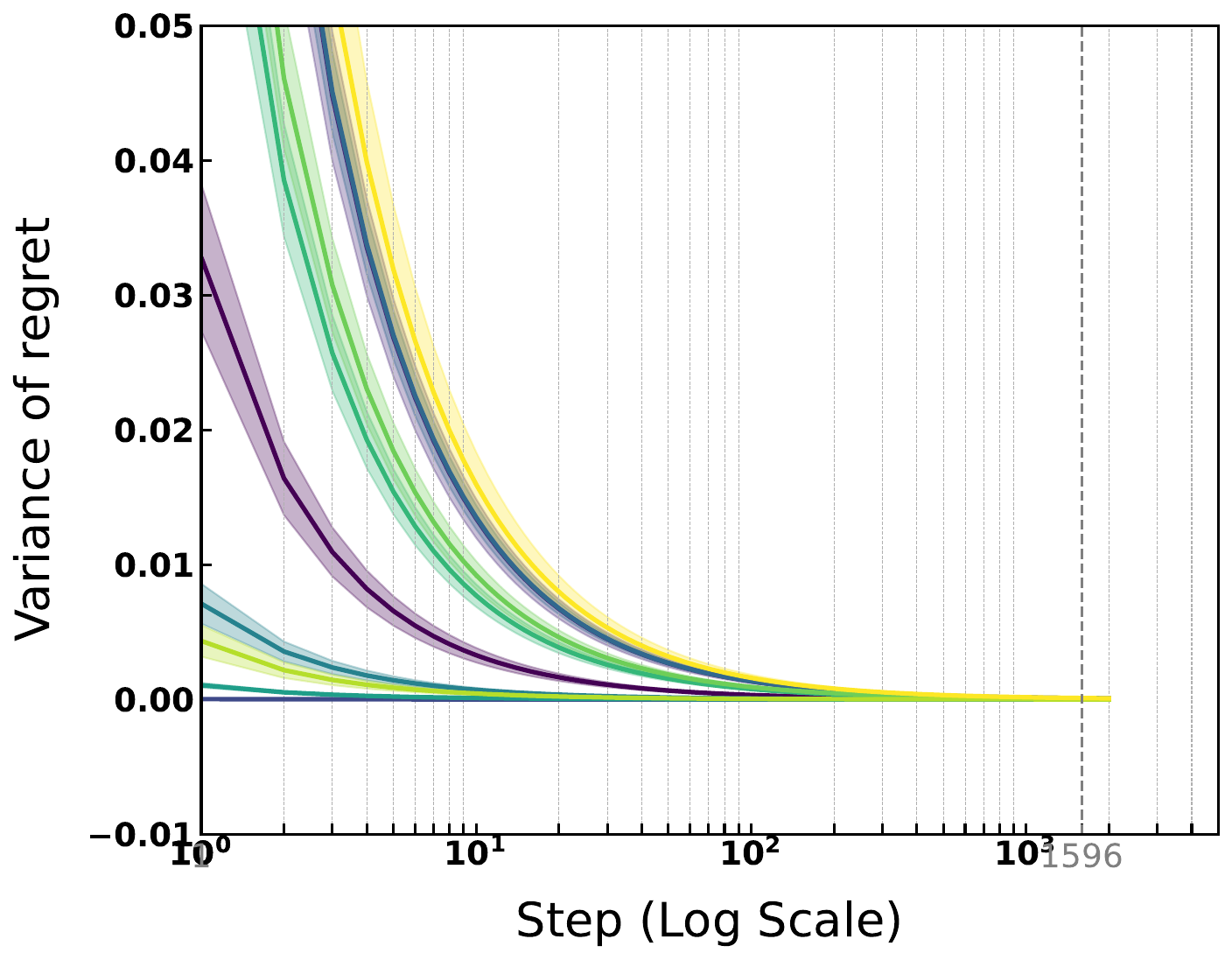}
    }
    \subfloat[PD, Norm($\sigma=0.1$)\label{fig:ws_pd_var0.1_vr}]{
        \includegraphics[width=0.23\textwidth]{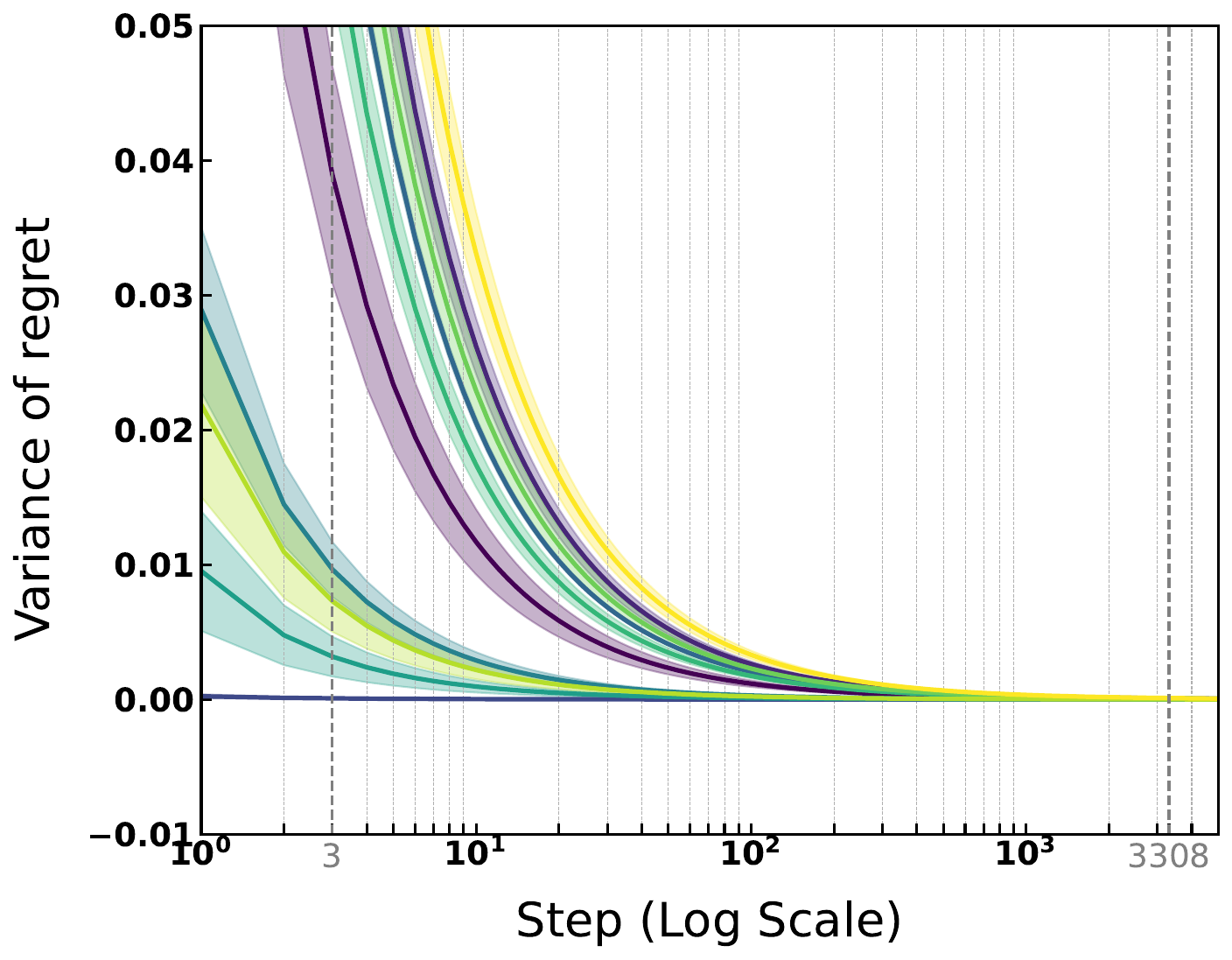}
    }
    \subfloat[RPS, Norm($\sigma=0.05$)\label{fig:ws_rps_var0.05_vr}]{
        \includegraphics[width=0.23\textwidth]{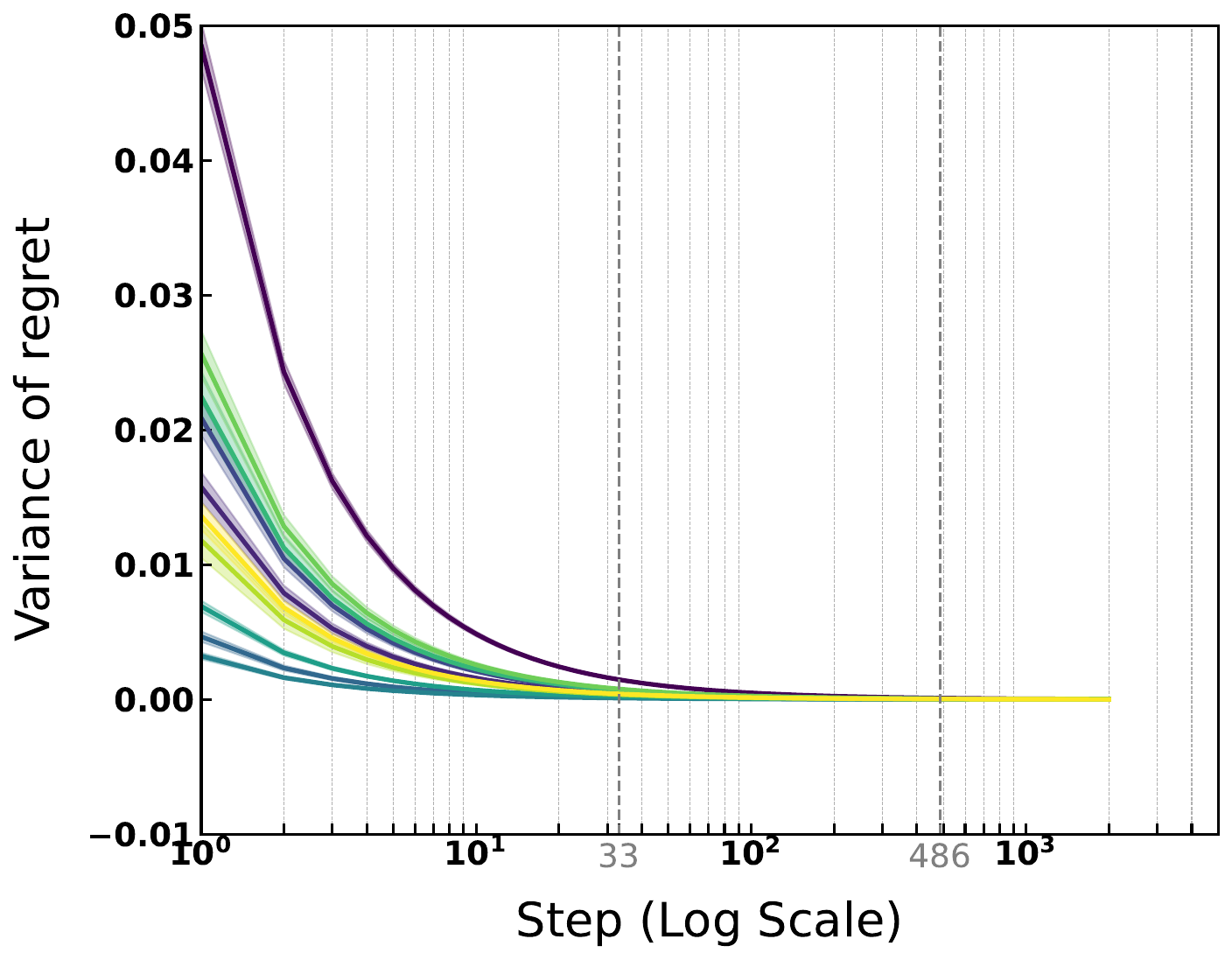}
    }
    \subfloat[RPS, Norm($\sigma=0.1$)\label{fig:ws_rps_var0.1_vr}]{
        \includegraphics[width=0.23\textwidth]{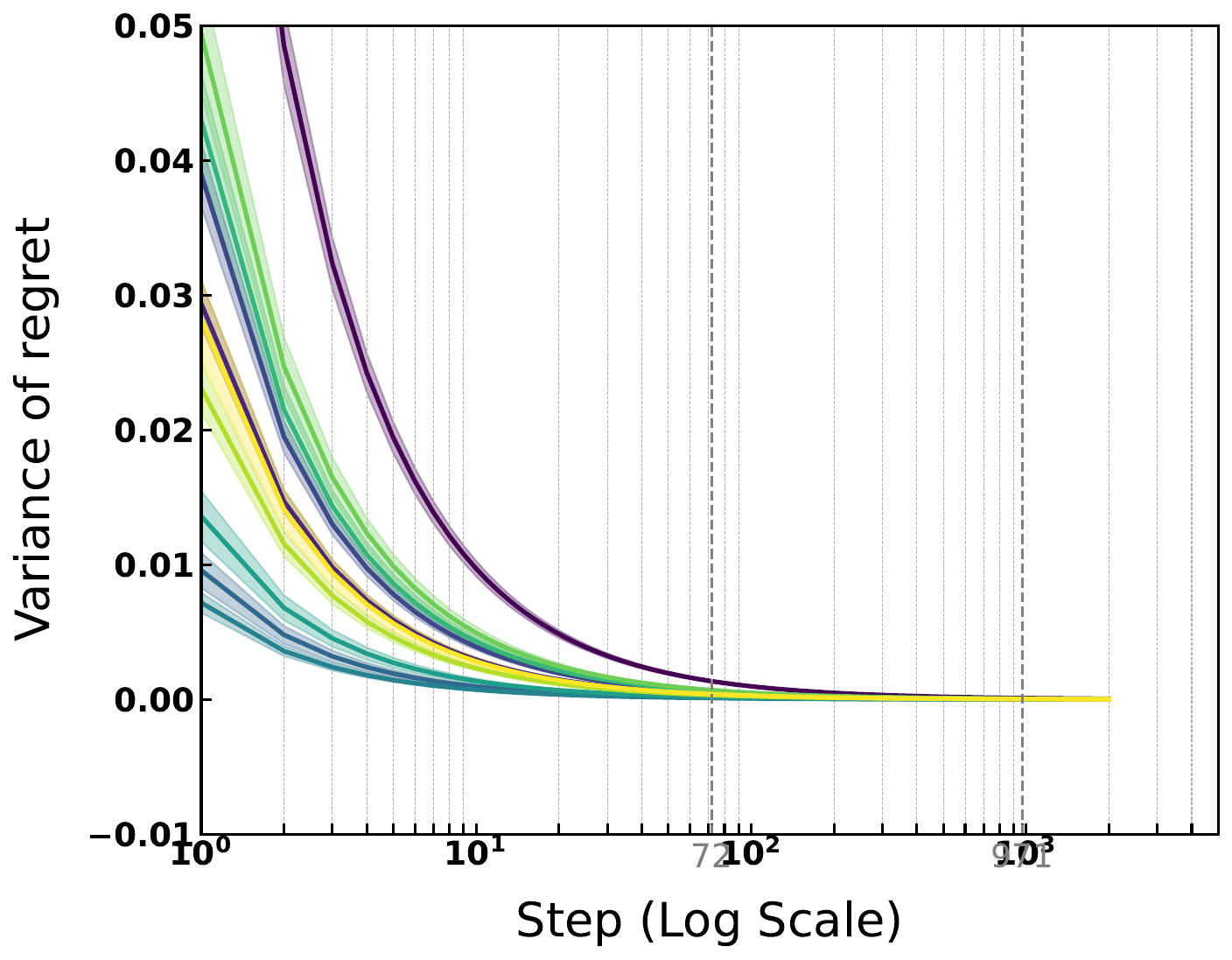}
    }
    \caption{Evolution of mean policy and regret variance in agent-based simulations on Watts-Strogatz networks for the Prisoner's Dilemma and Rock-Paper-Scissors games. Initial regrets are sampled from normal distributions with varying means, while initial normalized variances ($\sigma$) are identical. Each line represents a population, where mean policies (a-d) converge to the QRE and regret variances (e-h) decrease over time. Dashed lines mark the maximum steps across 10 simulations for heterogeneous populations to reach homogeneity in simulations. Notably, systems with higher initial regret variance consistently require more steps to achieve homogeneity.}
    \label{fig:multipopulation}
    \vspace{-0.1in}
\end{figure*}

\begin{figure*}[!htbp]
    \centering
    \subfloat[BoS step=1\label{fig:bos_s1}]{
        \includegraphics[width=0.23\textwidth, height=0.23\textwidth]{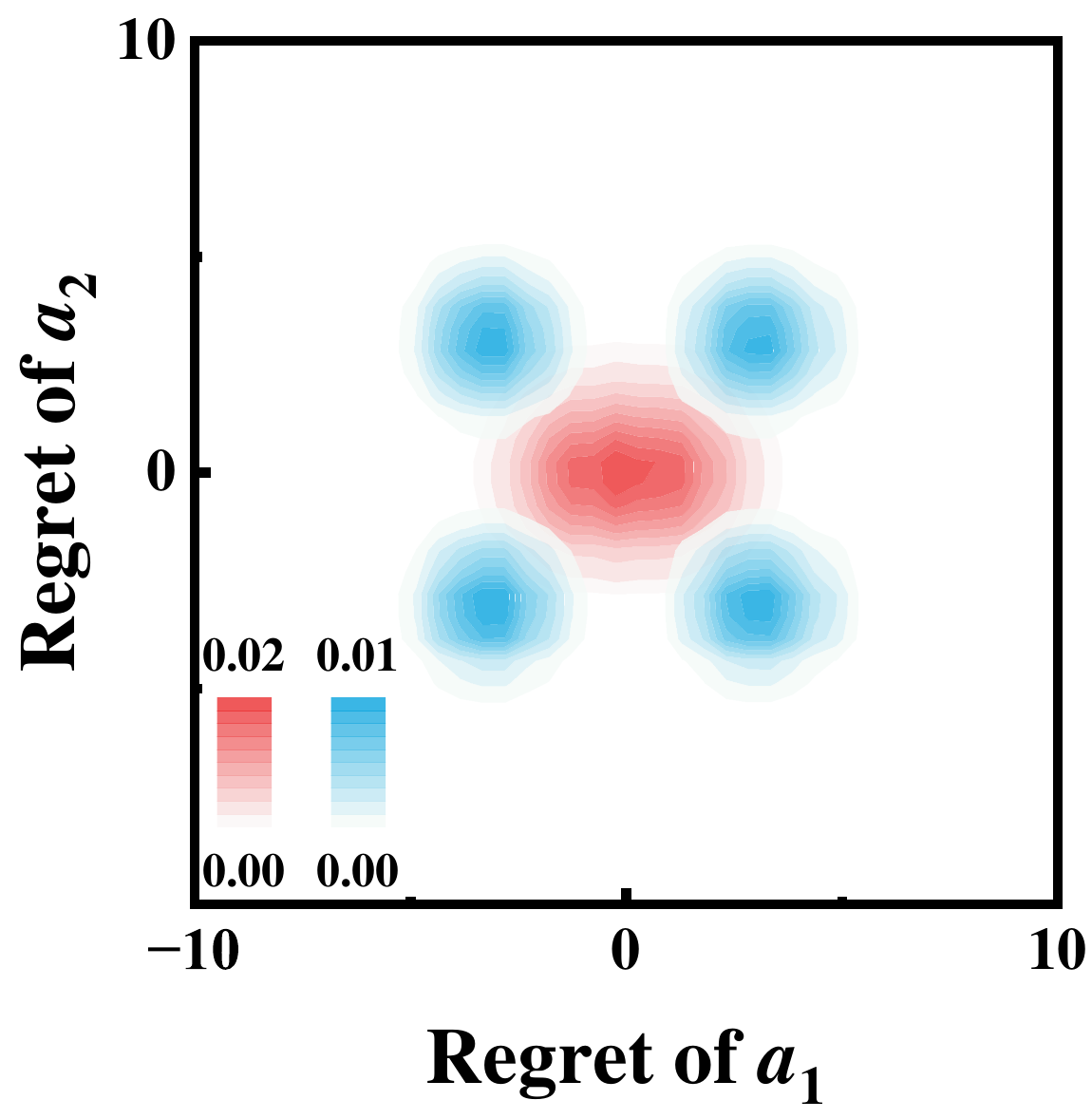}
    }
    \subfloat[BoS step=2\label{fig:bos_s2}]{
        \includegraphics[width=0.23\textwidth, height=0.23\textwidth]{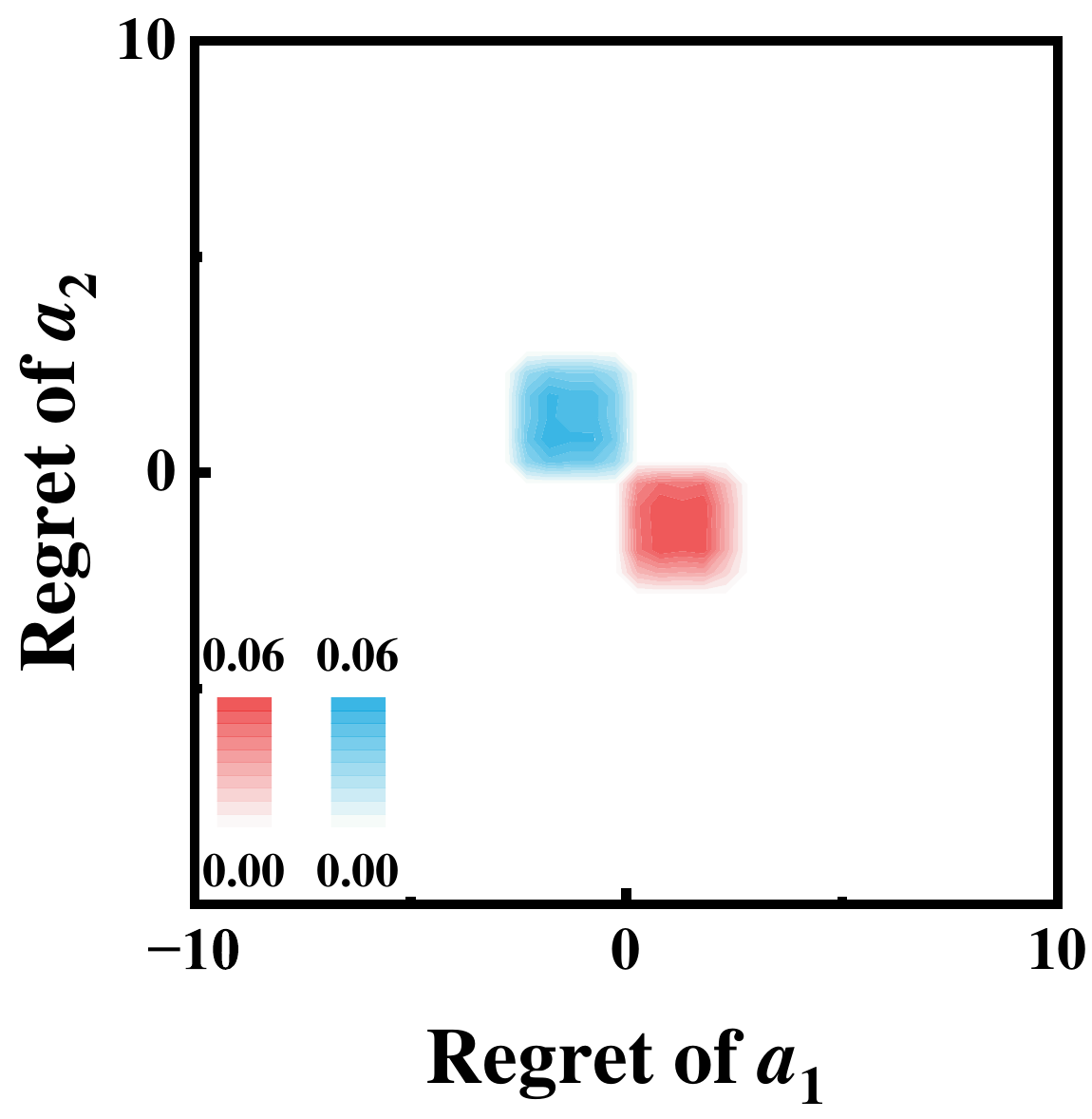}
    }
    \subfloat[BoS step=4\label{fig:bos_s4}]{
        \includegraphics[width=0.23\textwidth, height=0.23\textwidth]{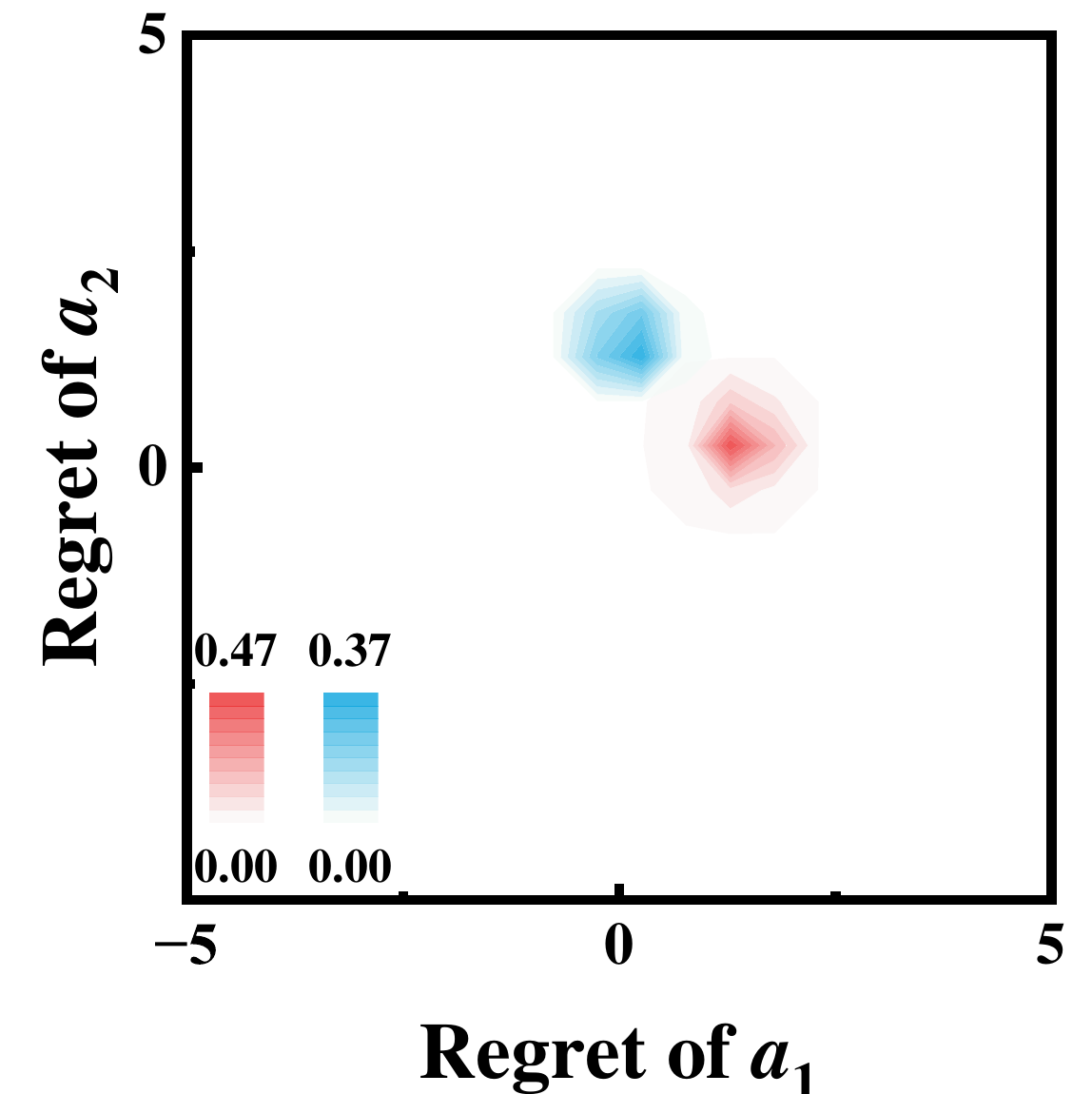}
    }
    \subfloat[BoS step=14\label{fig:bos_s14}]{
        \includegraphics[width=0.23\textwidth, height=0.23\textwidth]{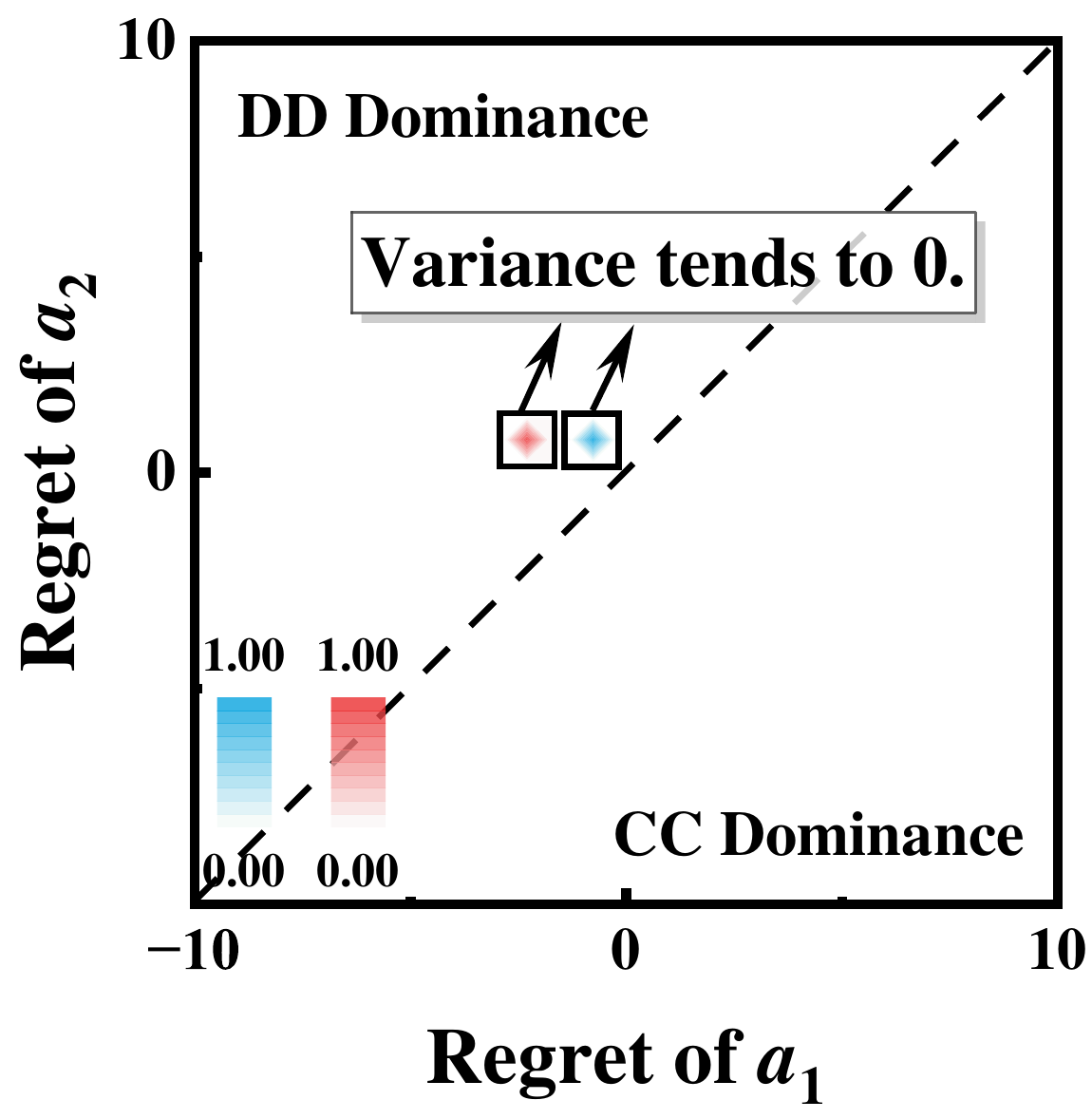}
    }
    \vfill
    \subfloat[PD step=1\label{fig:pd_s1}]{
        \includegraphics[width=0.23\textwidth, height=0.23\textwidth]{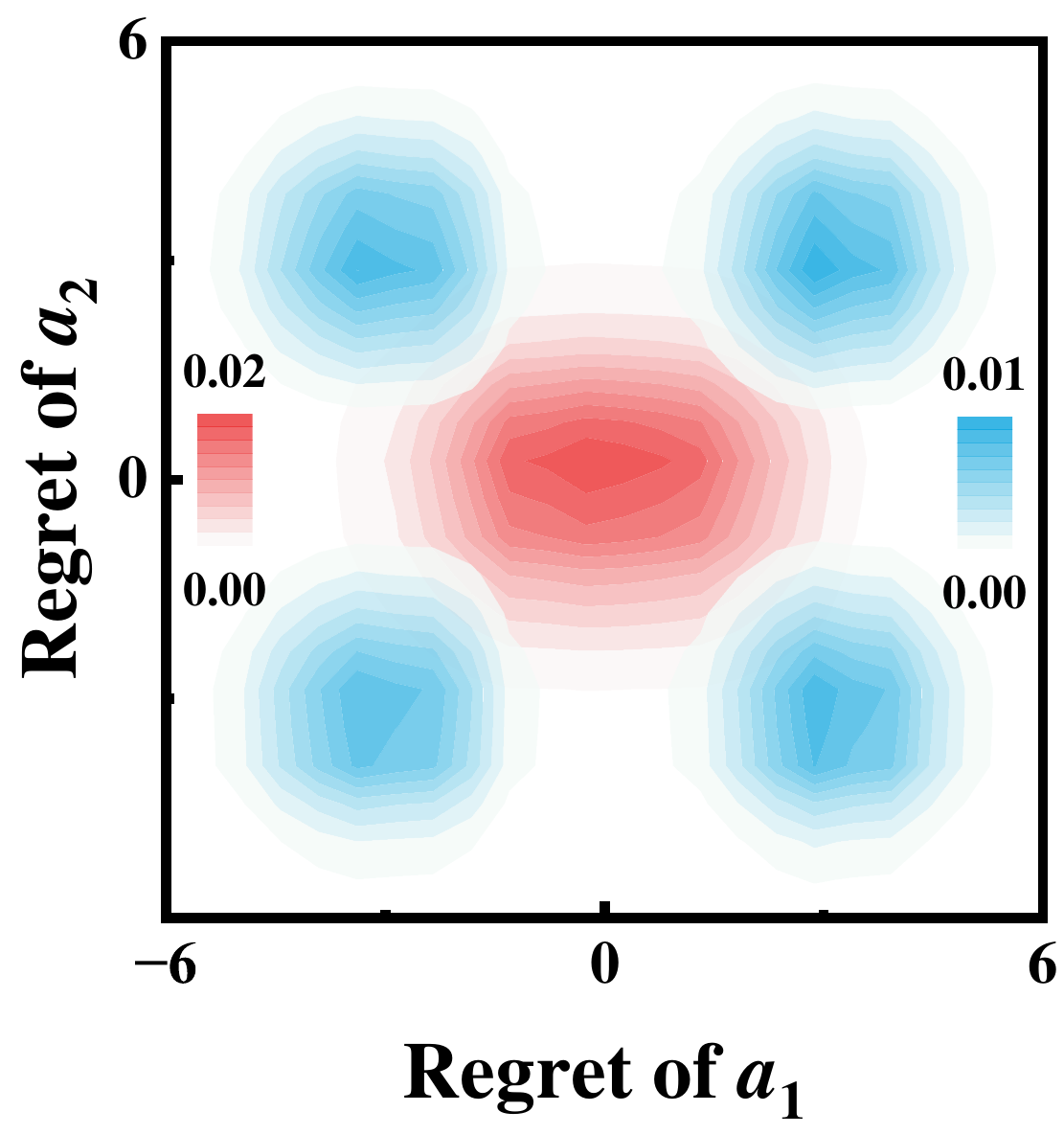}
    }
    \subfloat[PD step=2\label{fig:pd_s2}]{
        \includegraphics[width=0.23\textwidth, height=0.23\textwidth]{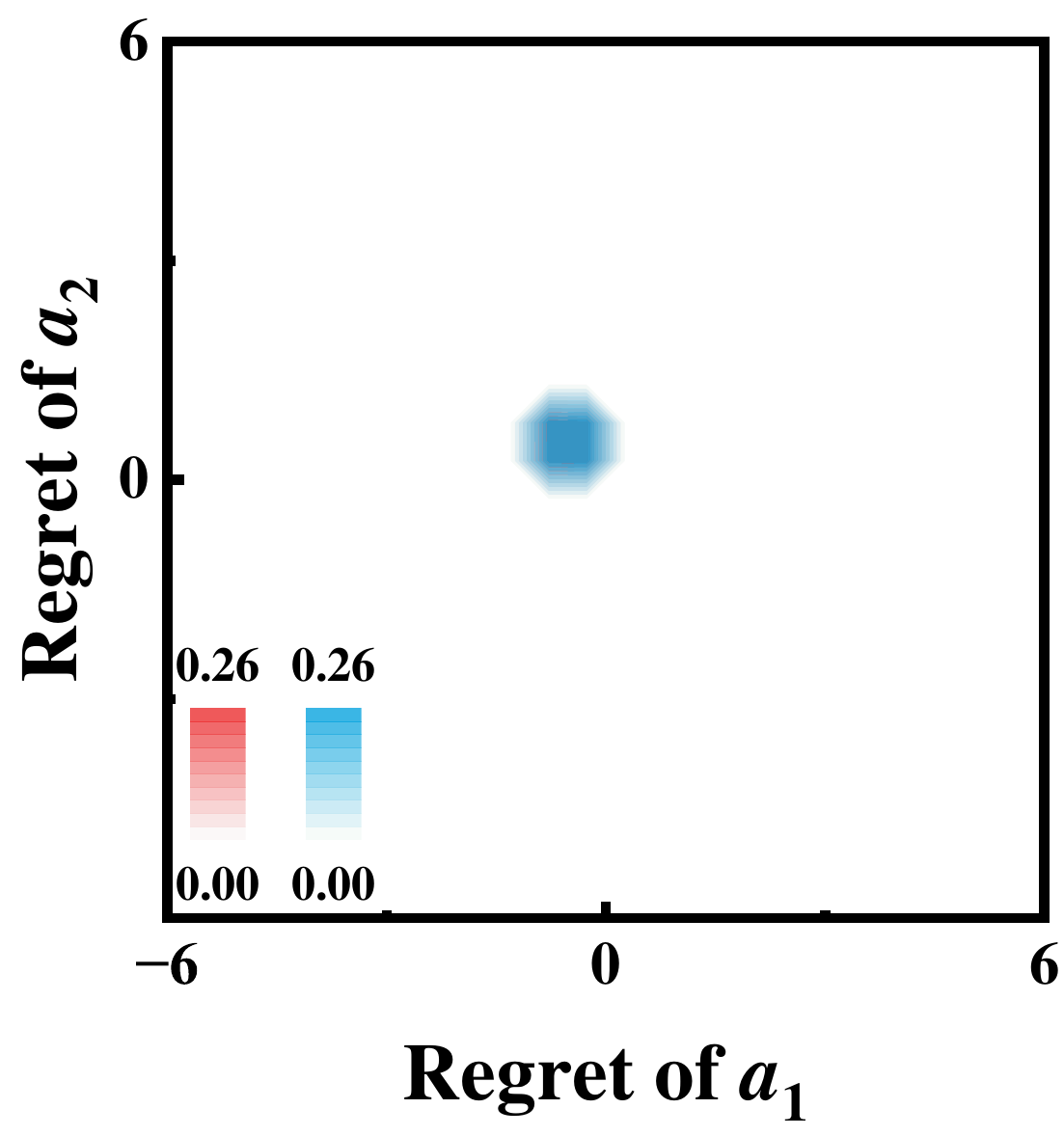}
    }
    \subfloat[PD step=4\label{fig:pd_s4}]{
        \includegraphics[width=0.23\textwidth, height=0.23\textwidth]{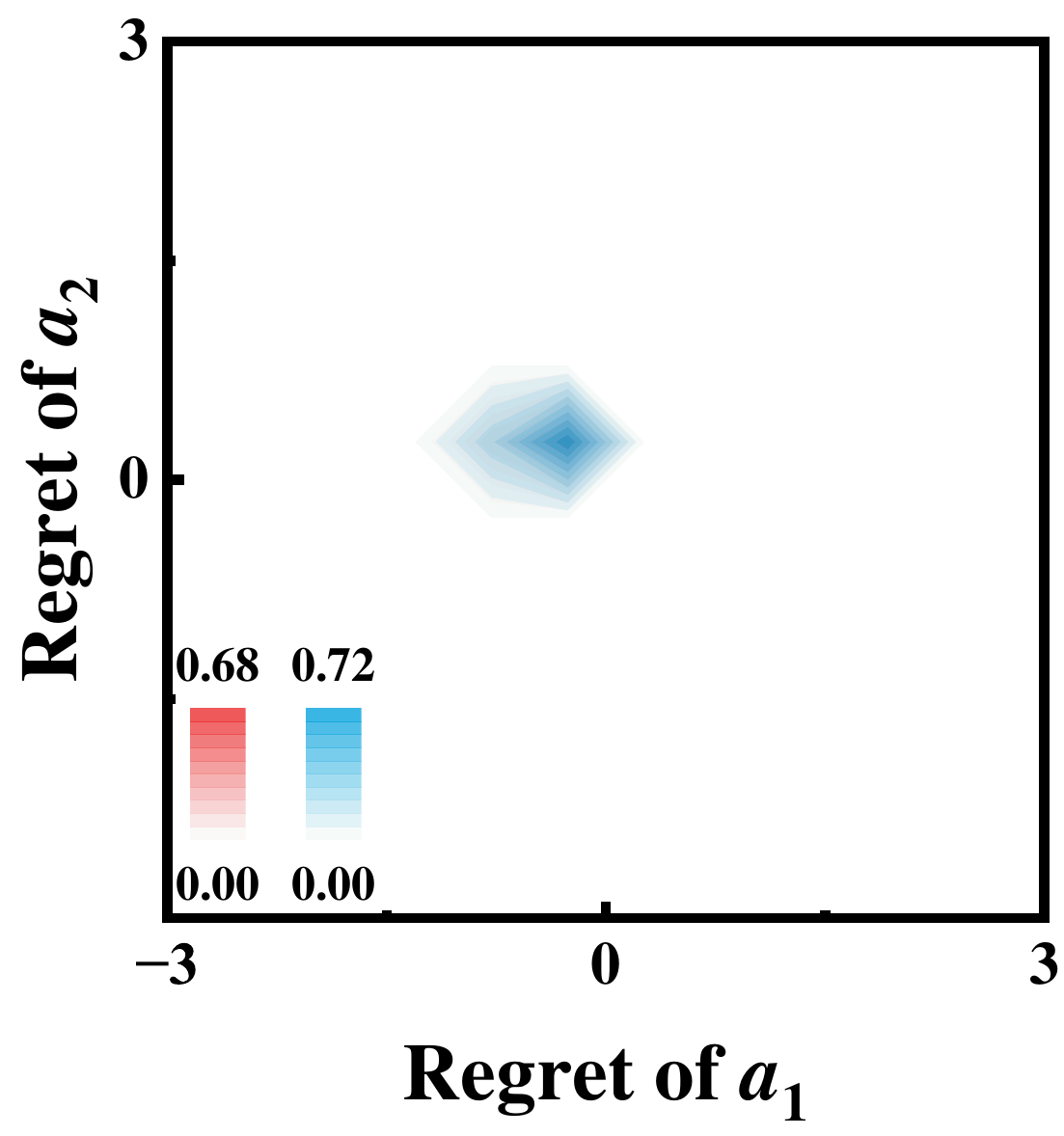}
    }
    \subfloat[PD step=20\label{fig:pd_s20}]{
        \includegraphics[width=0.23\textwidth, height=0.23\textwidth]{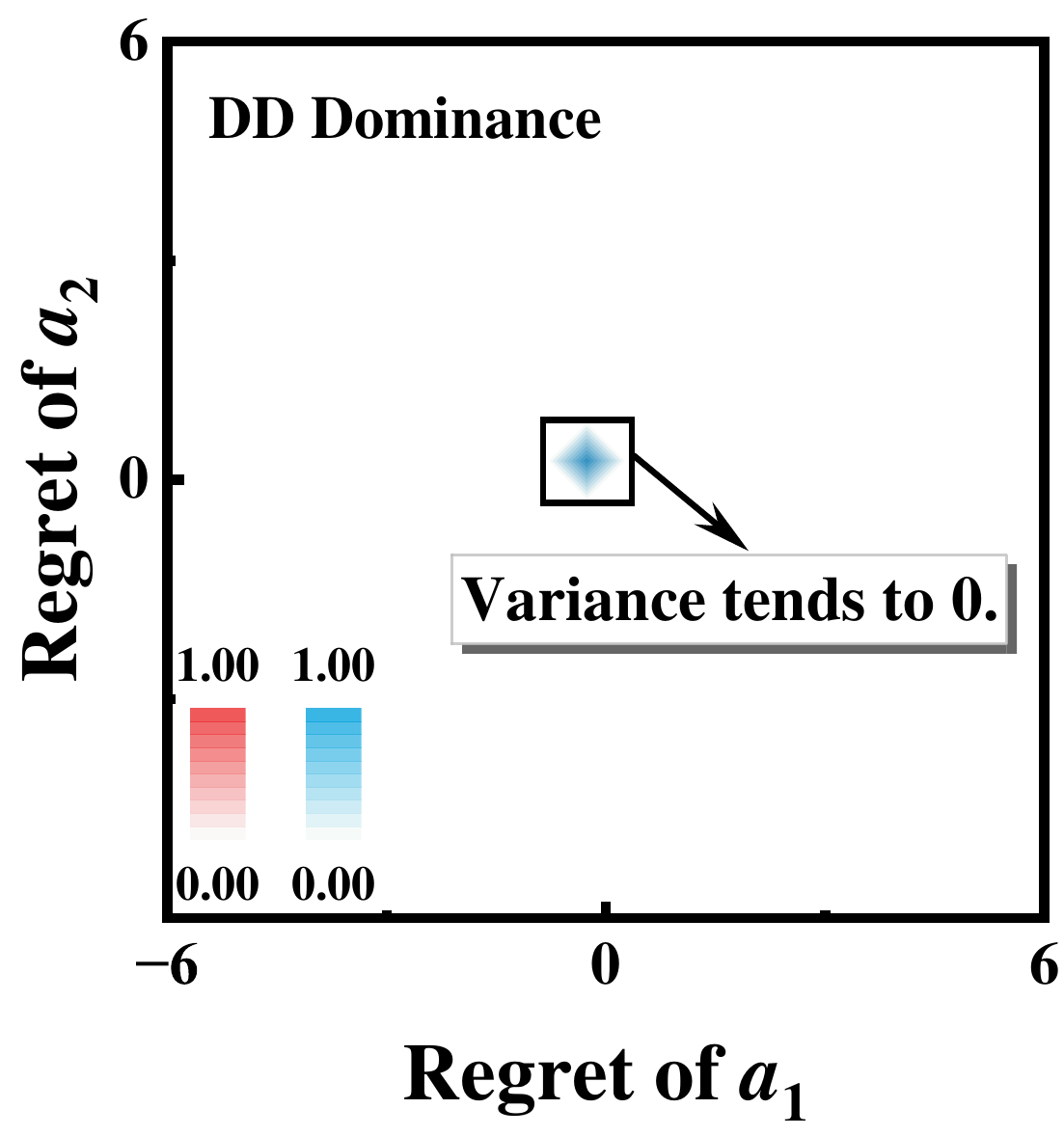}
    }
    \caption{Evolution of the distributions of regrets in two populations (red for population 1, and blue for population 2). The intensity of the color indicates the density (frequency) of agents. The upper panel corresponds to the BoS game, and the bottom panel corresponds to the PD game. It can be observed that agents' regrets become more concentrated over time, suggesting that the heterogeneity within the population gradually disappears.}
    \label{fig:varcontour}
    \vspace{-0.1in}
\end{figure*}

\section{Convergence Analysis in Weighted Zero-Sum and Potential Population Network Games}

Building upon the aforementioned results, in this section, we analyze the convergence properties in two important classes of games: weighted zero-sum PNGs and potential PNGs.
We shall show that under the softmax action selection function, agents' strategies will eventually converge to the Quantal Response Equilibria (QRE) of these games.
Prior to starting our convergence results, we define these games and the QRE.

\begin{definition}[Weighted Zero-sum PNG]\label{def:WZSPNG}
    A population network game is a weighted zero-sum (or competitive) if there exist positive constants $\omega_1,\ldots,\omega_n$ such that
    \begin{equation}
        \sum_{i\in V}\omega_iu_i(\mathbf{x})=\sum_{(i,j)\in E}\sum_{k\in v_i}(\omega_i\mathbf{x}_{i_k}^\top\mathbf{A}_{ij}\bar{\mathbf{x}}_j+\omega_j\bar{\mathbf{x}}_j^\top\mathbf{A}_{ji}\mathbf{x}_{i_k})=0,
    \end{equation}
    for all populations.
\end{definition}
\begin{definition}[Weighted Potential PNG]\label{def:WPPNG}
    A population network game is a weighted potential if there exists a potential function
    $U:\mathbf{x}\to\mathbb{R}$ and a vector of positive weights $w_1,\ldots,w_n$ such that
    \begin{equation}
        \begin{aligned}
             & u_i(\mathbf{x}_{i_k}, \bar{\mathbf{x}}_{-i})-u_i(\mathbf{x}_{i_k}^\prime, \bar{\mathbf{x}}_{-i})                                                      \\
             & =w_i[U(\mathbf{x}_i, \bar{\mathbf{x}}_{-i})-U(\mathbf{x}_{i_k}^\prime, \bar{\mathbf{x}}_{-i})],\ \mathbf{x}_{i_k},\mathbf{x}_{i_k}^\prime\in\Delta_i,
        \end{aligned}
    \end{equation}
    for any agent $k$ in each population $i$.
\end{definition}
We remark that weighted zero-sum PNGs include all the PNGs where the 2-player subgames are weighted zero-sum, and likewise, weighted potential PNGs include all the PNGs where the 2-player subgames admit a potential function.

On the other hand, QRE is a solution concept in game theory, which extends the standard Nash equilibrium to account for bounded rationality~\cite{mckelvey1995quantal}.
In a QRE, instead of directly choosing the action with the highest (observed) utility, individuals make decisions with the probability that is positively related to the payoff from that action. We generalize the QRE to PNG and define them as follows.

\begin{definition}[Quantal Response Equilibria of PNG]
    In a population network game, a quantal response equilibrium corresponds to the system state where for any agent $k$ in each population $i$, its policy satisfies
    \begin{equation}
        \mathbf{x}_{i_k}^*=\frac{\exp[\lambda\sum_{j\in V_i}\frac{1}{\vert V_j\vert}\mathbf{A}_{ij}\bar{\mathbf{x}}_j^*]}{\sum_{b\in S_i}\exp[\lambda\mathbf{e}_{ib}^\top\sum_{j\in V_i}\frac{1}{\vert V_j\vert}\mathbf{A}_{ij}\bar{\mathbf{x}}_j^*]},
        \label{eq:qre}
    \end{equation}
    where $\sum_{j\in V_i}\frac{1}{\vert V_j\vert}\mathbf{A}_{ij}\bar{\mathbf{x}}_j^*$ is the expected payoff of agent $k$ according to policy $\mathbf{x}_{i_k}^*$.
\end{definition}
Observe that for every agent in a population $i$, the right-hand side of \Cref{eq:qre} is the same.
Hence, our definition of QRE assumes that all the agents of a population will develop the same policy at the equilibrium state.
This is feasible under our concerned setting since \Cref{theorem:VarDecay} has indicated that the variance of regrets will always decay to zero, and consequently, agents will converge to the same policy.

We are now ready to present our convergence result:

\begin{theorem}[Convergence to QRE in Weighted Zero-sum and Weighted Potential PNG]\label{theorem:hete_convergence}\label{theorem:convergence}
    Agents' policies converge to a unique quantal response equilibrium in a weighted zero-sum population network game and to a compact connected set of quantal response equilibria in a weighted potential population network game.
\end{theorem}

\begin{proofsketch}
    Here is a streamlined explanation step by step.

    (1) \textbf{Foundation about Asymptotically Autonomous.}
    Given the time-varying nature of the system \eqref{eq:heto_dERdt}, directly obtaining convergence results is challenging. To address this, we apply the asymptotically autonomous theorem to approximate the limit behavior of \eqref{eq:heto_dERdt} through a simpler dynamics of its autonomous system. The definition of asymptotically autonomous is first given as follows.

    \begin{definition}[Asymptotically Autonomous \cite{markus1956asymptotically}]\label{def:asy}
        Consider generalized differential equations
        \begin{subequations}\label{eq:asy_auto}
            \begin{align}
                 & \dot{x}=f(t,x),\label{eq:asy_auto_a} \\
                 & \dot{y}=g(y),\label{eq:asy_auto_b}
            \end{align}
        \end{subequations}
        in $\mathbf{R}^n$. We say \eqref{eq:asy_auto_a} is asymptotically autonomous - with limit \eqref{eq:asy_auto_b} if
        \begin{equation}
            f(t,x)\to g(x),\ t\to\infty,
        \end{equation}
        locally uniformly in $x\in\mathbf{R}^n$.
    \end{definition}

    This principle states that if, over time, the dynamics of a time-varying system approaches that of the limit equation (an autonomous system), then the time-varying system is considered asymptotically autonomous. In other words, as time goes on, the behavior of the time-varying system will align with the autonomous system.

    (2) \textbf{Application to Regret Dynamics.} Building on \Cref{theorem:VarDecay}, which posits that the variance of regrets asymptotically tends to zero, we can establish the following lemma.

    \begin{lemma}\label{lemma:heto_dERdt}
        Let $\tau=\ln t$.
        For a system that initially has heterogeneous regrets, the time-varying mean regret dynamic \eqref{eq:heto_dERdt} is asymptotically autonomous with the limit equation:
        \begin{equation}\label{eq:rlimiteq}
            \frac{d\mathbf{R}_{i}}{d\tau}=\mathbf{r}_i-\mathbf{R}_i,
        \end{equation}
        which is equivalent to the regret dynamic for homogeneous systems after time reparameterization.
    \end{lemma}

    By reparameterizing time as $\tau=\ln t$, we conclude that the time-varying mean regret dynamic \eqref{eq:heto_dERdt} in a system with initially heterogeneous regrets approaches an autonomous system \eqref{eq:rlimiteq}.

    \textbf{(3) Transition to Policy Dynamic.} Then by the chain rule, we focus on the policy dynamics.
    \begin{corollary}[Mean Policy Dynamics]\label{corollary:heto_dExdt}

        For each population $i\in V$, the continuous-time dynamic of the mean policy $\bar{x}_{ia}$ can be described by
        \begin{equation}\label{eq:heto_dExdt}
            \frac{\partial \bar{x}_{ia}}{\partial t}
            =\frac{\partial \bar{x}_{ia}}{\partial \bar{R}_{ia}}\frac{\partial \bar{R}_{ia}}{\partial \tau}+\sum_{a'\neq a}\frac{\partial \bar{x}_{ia}}{\partial \bar{R}_{ia'}}\frac{\partial \bar{R}_{ia'}}{\partial \tau}
        \end{equation}
        where $a\in S_i$ and $\bar{R}_{ia}$ is the element of mean regret $\bar{\mathbf{R}}_i$.

    \end{corollary}

    \begin{lemma}\label{lemma:dxdt}
        The continuous-time dynamic of population policy $x_{ia}$ is asymptotically autonomous with the limit equation:
        \begin{equation}\label{eq:xlimiteq}
            \frac{dx_{ia}}{d\tau}=x_{ia}[\lambda u_{i}(a,\bar{\mathbf{x}}_{-i})-\lambda\mathbf{x}_i^\top\mathbf{u}_i(\bar{\mathbf{x}}_{-i})-\ln x_{ia}+\mathbf{x}_{i}^\top\ln \mathbf{x}_{i}]
        \end{equation}
        where $\mathbf{u}_i(\bar{\mathbf{x}}_{-i})=\frac{1}{\vert V_i\vert}\sum_{j\in V_i}\mathbf{A}_{ij}\bar{\mathbf{x}}_{j}$.
    \end{lemma}
    We observe that \Cref{eq:xlimiteq} (the limit equation of policy dynamics \eqref{eq:heto_dExdt}) is reduced to a smooth Q-learning dynamic with the learning rate $\alpha=1$. This smooth Q-learning dynamic, initially introduced by Tuyls et al. \cite{2003tuylsSelectionmutationModelQlearning}, has garnered significant interest within the community \cite{hennes2009state,kaisers2010frequency,kianercy2012dynamics,panozzo2014evolutionary}. Applying asymptotically autonomous theorem (\Cref{def:asy}), we conclude that the limit behavior of smooth regret matching tends to the behavior of Q-learning over time. This allows us to leverage convergence results to gain insights into the dynamics of the challenging time-varying system.

    \textbf{(4) Convergence to QRE.} Prior to demonstrating the convergence result, we introduce the convergence proof for the n-player game as established by Leonardos et al.

    \begin{lemma}[Leonardos et al.~\cite{nips2021leonardos} Theorem 4.1 and~\cite{2021leonardosExplorationExploitationMultiAgentLearning} Theorem 3]\label{lemma:leonardos}
        Given a positive exploration rate $\lambda>0$. In an n-player weighted zero-sum network game, the sequence of play generated by the SQL dynamic converges to a unique QRE exponentially fast.
        In an n-player weighted potential game, the sequence of play generated by the SQL dynamic converges to a compact connected set of QRE.
    \end{lemma}

    Based on \Cref{theorem:VarDecay}, as time approaches infinity, a heterogeneous population may be viewed as a single genetic agent, allowing us to align this asymptotic dynamic of homogeneous n-population network games with the dynamics observed in n-player games. To figure out if our policy dynamics' long-term behavior aligns seamlessly with the smooth Q-Learning dynamic, we draw upon Thieme~\cite{1992thiemeConvergenceResultsPoincareBendixson}'s seminal results on the connection between the limit behavior of an asymptotically autonomous system $\Phi$ and its limit equation $\Theta$.

    \begin{lemma}[Thieme~\cite{1992thiemeConvergenceResultsPoincareBendixson} Theorem 4.2]
        Given a metric space $(X,d)$, assuming that the equilibria of $\Theta$ are isolated compact $\Theta$-invariant subsets of $X$. Further, the $\omega$-$\Theta$-limit set of any pre-compact $\Theta$-orbit contains a $\Theta$-equilibrium, and the point $(s,x), s\geq t_0, x\in X$, have a pre-compact $\Phi$-orbit. Then the following alternative holds 1) $\Phi(t,s,x)\to e,t\to\infty$, for some $\Theta$-equilibrium $e$, and 2) the $\omega$-$\Phi$-limit set of $(s,x)$ contains finitely many $\Theta$-equilibria which are chained to each other in a cyclic way.
    \end{lemma}

    This means that if the equilibria of the limit policy dynamics $\Theta$~\eqref{eq:xlimiteq} are isolated and any solution of $\Theta$~\eqref{eq:xlimiteq} converges to one of them, every solution of $\Phi$~\eqref{eq:heto_dExdt} converges to an equilibrium of $\Theta$~\eqref{eq:xlimiteq} as well. With the application of Thieme's convergence results for asymptotically
    autonomous systems, every solution of our policy dynamics also converges to a unique QRE in a weighted zero-sum game and to a compact connected set of QRE in a weighted potential game. This concludes the proof.
\end{proofsketch}

Theorem \ref{theorem:convergence} guarantees that the real-time system behavior converges to QRE, regardless of how heterogeneous the system initially is.
Moreover, \Cref{theorem:convergence} provides a tractable approach to the equilibrium selection problem in weighted zero-sum and potential population network games.

We remark that in the realm of regret matching algorithms, Theorems 1, 2, and 3 exhibit a remarkable degree of universality, being applicable not only to the traditional regret matching framework but also to its refined iterations, including regret matching plus and predictive regret matching plus. Theorem \ref{theorem:convergence}, however, deviates from this pattern by mandating the adoption of a Boltzmann-based approach for policy updates, distinguishing it from the others. This specificity opens avenues for intriguing future research endeavors.

\section{Experiments}\label{sec:experiments}

In this section, we present our empirical study by comparing the behavior of our real-time population dynamics model with behaviors obtained from agent-based simulations.
For the model, we employ the moment closure technique to solve the intractable continuity equations depicted in \Cref{eq:PDF_R1}. For the agent-based simulations,
we perform multiple independent simulations with different regret initializations to provide evidence for our convergence results (\Cref{theorem:convergence}). Additionally, we use two concrete examples to validate the phenomena of vanishing heterogeneity (\Cref{theorem:VarDecay}) and study its impact on equilibrium selection.

\subsection{Experimental Settings}
We consider six typical 2-population network games to experiment on three weighted zero-sum games, i.e., presidential election~(PE), asymmetric matching pennies~(AMPs) and rock-paper-scissor~(RPS), and three weighted potential games, i.e., prisoner's dilemma~(PD), stag hunt~(SH) and the battle of the sexes~(BoS). The payoff bimatrices of these games are shown in \Cref{tab:payoff}.

The equilibria are diverse in these games.
Specifically, the PE, AMPs, and RPS games are weighted zero-sum games without pure Nash equilibrium.
However, each game has a specific mixed Nash equilibrium.
The mixed Nash equilibrium of PE game is $\left(\frac{3}{7}E+\frac{4}{7}S,\frac{2}{7}M+\frac{5}{7}T\right)$,
that of AMPs game is $\left(\frac{1}{3}H+\frac{2}{3}T,\frac{2}{3}H+\frac{1}{3}T\right)$,
and that of RPS game is $\left(\frac{1}{3}R+\frac{1}{3}P+\frac{1}{3}S,\frac{1}{3}R+\frac{1}{3}P+\frac{1}{3}S\right)$.
By contrast, the finite (weighted) potential games always have at least one pure Nash equilibrium~\cite{cheng2014finite}. For instance, the PD game has a unique  Nash equilibrium $(D, D)$. The SH game has two symmetric pure Nash equilibria, $(S, S)$ and $(H, H)$, whereas $(S, S)$ is the payoff-dominant equilibrium that maximizes social welfare.
There is also a symmetric mixed Nash equilibrium for the SH game, i.e., $\left(\frac{2}{3}S+\frac{1}{3}H,\frac{2}{3}S+\frac{1}{3}H\right)$.
For the BoS game, there exists two symmetric pure Nash equilibria, $(F, F)$ and $(B, B)$, and an asymmetric mixed Nash equilibrium $\left(\frac{2}{3}F+\frac{1}{3}B,\frac{1}{3}F+\frac{2}{3}B\right)$.

For each setting, we consider $100$ agents in each population and run 50 independent simulations to smooth out the randomness. the temperature for action selection is set to $1$.

\subsection{Convergence to Quantal Response Equilibriua}

Figure \ref{fig:sim_2population} illustrates the QRE sets through black dotted lines for each game along with a comparison of the expected policy, presented through dash lines obtained from our model, and the mean policy of the populations, represented by solid lines. Without loss of generality, it is assumed that the initial regret values for each action of each population are generated from a positive Normal distribution. Additional information about this is provided in supplementary materials.

Our model demonstrates a remarkable ability to effectively capture the distinctive patterns of evolution exhibited by agent populations involved in various kinds of games. Particularly noteworthy is our model indeed provides a description of the learning dynamics in the equilibrium selection in the context of weighted zero-sum games and weighted potential games.

\subsubsection{Weighted Zero-sum Games}
In weighted zero-sum games represented in \Cref{fig:sim_2population}(a)-(c), the mean policy of each population indeed converges to the unique QRE, as predicted in \Cref{theorem:convergence}. It is worth noting that in the AMPs and RSP games, the corresponding unique QRE is asymmetric due to the asymmetry of the game. As a result, the two populations will eventually converge to different strategies.

\subsubsection{Weighted Potential Games}
For weighted potential games presented in \Cref{fig:sim_2population}(d)-(f), it is evident that each mean policy converges to the QRE set, as predicted in \Cref{theorem:convergence}. Specifically, in the PD and SH games, the mean policy converges to the corresponding unique QRE.
However, in the BoS game, there are three QREs.
Though the initial distribution of regrets may lead to different outcomes, the policy always converges to one of the QRE.

\subsubsection{Multi-population in Different Networks}
Furthermore, we examine the evolution of the mean strategy and regret variance for each population in Watts-Strogatz networks within the context of prisoner's dilemma (a weighted potential game) and rock-paper-scissors (a weighted zero-sum game), as shown in \Cref{fig:multipopulation}. Initial regrets are sampled from normal distributions with different means of regret across populations, while the initial normalized variances of regret (denoted by $\sigma$) are the exactly same. Notably, when comparing \Cref{fig:multipopulation}(e) and (f) or (g) and (h), despite starting from the same mean policy,
we observe that the systems with larger initial variance in regret always achieve to homogeneity more slowly, requiring more time steps for the variance to approach zero when initial variance is larger.

\subsection{Vanishing Heterogeneity and its Impact on Equilibrium Selection}

Figure \ref{fig:varcontour} snapshots the evolution of joint regret probability distributions for two populations in BoS and PD games. In each game, the initial regret values for each action of each population are generated from a multimodal normal distribution.

In the BoS game, we can observe that the initial joint probability distribution for regrets among the population is relatively heterogeneous, with a maximum probability of joint initial regret being only 0.02.
As time evolves, agents learn and adjust their policies, which leads to a more concentrated distribution and a decrease in the variance of joint probability for regrets.
Eventually, by step 14, the distribution reaches a point where the maximum probability increases to 1 for both populations. In other words, the variance dropped to zero, implying that the system heterogeneity vanished and that all agents in the population had the same regret.
This supports the \Cref{theorem:VarDecay} that all agents among a population will ultimately adopt the same policy, despite starting with varying policies.

However, how the initial regret distribution determines which QRE will converge in those games with more than one QRE, such as the BoS game, naturally raises the problem of equilibrium selection. We note that when the variance decays to 0, the system has completed the transition from a state of heterogeneity to that of homogeneity. We then shift our attention to the homogeneous populations and model the dynamic for homogeneous populations using \Cref{eq:rlimiteq}.
This is illustrated in \Cref{fig:bos_s14}, which shows the attractive region for the different QRE.
From this, the equilibrium to which the system converges can be identified.

In the PD game, the variance also decreases to 0 as the game proceeds. Nevertheless, the PD game only has one QRE. Thus, as shown in \Cref{fig:pd_s1}-\ref{fig:pd_s20}, the probability density also concentrates at a single point, and all agents ultimately converge to that unique QRE.

\section{Discussion}
In this paper, we study a variant of regret matching and a common large-scale multi-agent setting that allows for heterogeneous agents, each starting with a generally different regret value. We start by modeling the dynamic with a partial differential equation for PNG in \Cref{theorem:regretdynamics}, which allows us to precisely examine the impact of heterogeneous regrets on system evolution. Then applying moment closure, we prove in \Cref{theorem:heto_dERdt} that the mean regret dynamic is affected not only by the mean regret itself but also by the variance of regrets, which quantifies the degree of heterogeneity in the system. We show in \Cref{theorem:VarDecay} that the system heterogeneity vanishes over time. In \Cref{theorem:convergence}, we establish the convergence of smooth regret matching to a unique quantal response equilibrium in weighted zero-sum population network games and to a compact connected set of quantal response equilibria in weighted potential population network games. The experiments through numerical simulations and agent-based simulations on six typical games and different initial settings of regret values validate our theoretical findings.
Our results have direct implications for real-world adaptive systems. In traffic networks, they explain how decentralized route choices synchronize over time, reducing congestion. Similarly, the convergence guarantees support predictive mechanisms for market stability. These insights provide a foundation for robust, decentralized decision-making in complex systems.

Looking forward, our theoretical framework offers strong potential for extension to other game types as long as these involve two-player subgame scenarios. One possible direction would be extensive-form games since each extensive-form game can be uniquely mapped to its corresponding normal-form game through equivalent strategies \cite{cressman2003evolutionary}. However, extending our framework to multiplayer games with more than two players introduces significant complexity, particularly due to nonlinear payoff functions and the challenge of modeling interactions between multiple agents. Addressing this will likely require advancements in novel techniques to account for nonlinearity in multiplayer systems. Recent advancements, such as those discussed in \cite{wang2024evolutionary}, may offer valuable insights and inspiration for addressing these challenges in future research.
Another interesting avenue would be to explore the impact of dynamic network structures on regret dynamics.
Incorporating such dynamics introduces additional challenges, including time-varying opponent distributions, fluctuating interaction neighborhoods, and non-stationary learning signals. Addressing these challenges would require fundamentally new mathematical techniques, such as stochastic network evolution models or time-dependent differential equations, which go beyond the scope of our current analysis.
Despite these open challenges, we believe our findings lay the groundwork for future studies and can serve as a starting point for extending regret-based learning to more flexible and adaptive environments.

\bibliographystyle{IEEEtran}
\bibliography{TNNLS-2024-P-34131}

\begin{thebibliography}{10}
\providecommand{\url}[1]{#1}
\csname url@samestyle\endcsname
\providecommand{\newblock}{\relax}
\providecommand{\bibinfo}[2]{#2}
\providecommand{\BIBentrySTDinterwordspacing}{\spaceskip=0pt\relax}
\providecommand{\BIBentryALTinterwordstretchfactor}{4}
\providecommand{\BIBentryALTinterwordspacing}{\spaceskip=\fontdimen2\font plus
\BIBentryALTinterwordstretchfactor\fontdimen3\font minus \fontdimen4\font\relax}
\providecommand{\BIBforeignlanguage}[2]{{%
\expandafter\ifx\csname l@#1\endcsname\relax
\typeout{** WARNING: IEEEtran.bst: No hyphenation pattern has been}%
\typeout{** loaded for the language `#1'. Using the pattern for}%
\typeout{** the default language instead.}%
\else
\language=\csname l@#1\endcsname
\fi
#2}}
\providecommand{\BIBdecl}{\relax}
\BIBdecl

\bibitem{zhang2018data}
Q.~Zhang and D.~Zhao, ``Data-based reinforcement learning for nonzero-sum games with unknown drift dynamics,'' \emph{IEEE transactions on cybernetics}, vol.~49, no.~8, pp. 2874--2885, 2018.

\bibitem{zhang2020rnn}
J.~Zhang, L.~Jin, and L.~Cheng, ``Rnn for perturbed manipulability optimization of manipulators based on a distributed scheme: A game-theoretic perspective,'' \emph{IEEE Transactions on Neural Networks and Learning Systems}, vol.~31, no.~12, pp. 5116--5126, 2020.

\bibitem{li2022online}
M.~Li, S.~Chen, Y.~Shen, G.~Liu, I.~W. Tsang, and Y.~Zhang, ``Online multi-agent forecasting with interpretable collaborative graph neural networks,'' \emph{IEEE Transactions on Neural Networks and Learning Systems}, 2022.

\bibitem{heusel2017gans}
M.~Heusel, H.~Ramsauer, T.~Unterthiner, B.~Nessler, and S.~Hochreiter, ``Gans trained by a two time-scale update rule converge to a local nash equilibrium,'' \emph{Advances in neural information processing systems}, vol.~30, 2017.

\bibitem{tembine2019deep}
H.~Tembine, ``Deep learning meets game theory: Bregman-based algorithms for interactive deep generative adversarial networks,'' \emph{IEEE transactions on cybernetics}, vol.~50, no.~3, pp. 1132--1145, 2019.

\bibitem{jiang20233d}
Z.~Jiang, G.~Lu, X.~Liang, J.~Zhu, W.~Zhang, X.~Chang, and H.~Xu, ``3d-togo: Towards text-guided cross-category 3d object generation,'' in \emph{Proceedings of the AAAI Conference on Artificial Intelligence}, vol.~37, no.~1, 2023, pp. 1051--1059.

\bibitem{jiang2023dynamic}
Z.~Jiang, C.~Li, X.~Chang, L.~Chen, J.~Zhu, and Y.~Yang, ``Dynamic slimmable denoising network,'' \emph{IEEE Transactions on Image Processing}, vol.~32, pp. 1583--1598, 2023.

\bibitem{hannan1957approximation}
J.~Hannan, ``Approximation to bayes risk in repeated play,'' \emph{Contributions to the Theory of Games}, vol.~3, no.~2, pp. 97--139, 1957.

\bibitem{blum2007external}
A.~Blum and Y.~Mansour, ``From external to internal regret.'' \emph{Journal of Machine Learning Research}, vol.~8, no.~6, 2007.

\bibitem{brown2018superhuman}
N.~Brown and T.~Sandholm, ``Superhuman ai for heads-up no-limit poker: Libratus beats top professionals,'' \emph{Science}, vol. 359, no. 6374, pp. 418--424, 2018.

\bibitem{brown2019superhuman}
------, ``Superhuman ai for multiplayer poker,'' \emph{Science}, vol. 365, no. 6456, pp. 885--890, 2019.

\bibitem{2000hartSimpleAdaptiveProcedure}
S.~Hart and A.~{Mas-Colell}, ``\BIBforeignlanguage{en}{A {{Simple Adaptive Procedure Leading}} to {{Correlated Equilibrium}}},'' \emph{\BIBforeignlanguage{en}{Econometrica}}, vol.~68, no.~5, pp. 1127--1150, Sep. 2000.

\bibitem{hart2013simple}
S.~Hart and A.~Mas-Colell, \emph{Simple adaptive strategies: from regret-matching to uncoupled dynamics}.\hskip 1em plus 0.5em minus 0.4em\relax World Scientific, 2013, vol.~4.

\bibitem{zhang2021multi}
K.~Zhang, Z.~Yang, and T.~Ba{\c{s}}ar, ``Multi-agent reinforcement learning: A selective overview of theories and algorithms,'' \emph{Handbook of Reinforcement Learning and Control}, pp. 321--384, 2021.

\bibitem{zhang2023dynamics}
T.~Zhang, Z.~Lin, Y.~Wang, D.~Ye, Q.~Fu, W.~Yang, X.~Wang, B.~Liang, B.~Yuan, and X.~Li, ``Dynamics-adaptive continual reinforcement learning via progressive contextualization,'' \emph{IEEE Transactions on Neural Networks and Learning Systems}, 2023.

\bibitem{bailey2019fast}
J.~Bailey and G.~Piliouras, ``Fast and furious learning in zero-sum games: Vanishing regret with non-vanishing step sizes,'' \emph{Advances in Neural Information Processing Systems}, vol.~32, 2019.

\bibitem{zhu2020online}
Y.~Zhu and D.~Zhao, ``Online minimax q network learning for two-player zero-sum markov games,'' \emph{IEEE Transactions on Neural Networks and Learning Systems}, vol.~33, no.~3, pp. 1228--1241, 2020.

\bibitem{chen2022online}
L.~Chen, X.~Liang, Y.~Feng, L.~Zhang, J.~Yang, and Z.~Liu, ``Online intention recognition with incomplete information based on a weighted contrastive predictive coding model in wargame,'' \emph{IEEE Transactions on Neural Networks and Learning Systems}, 2022.

\bibitem{piliouras2022evolutionary}
G.~Piliouras, M.~Rowland, S.~Omidshafiei, R.~Elie, D.~Hennes, J.~Connor, and K.~Tuyls, ``Evolutionary dynamics and phi-regret minimization in games,'' \emph{Journal of Artificial Intelligence Research}, vol.~74, pp. 1125--1158, 2022.

\bibitem{chotibut2020route}
T.~Chotibut, F.~Falniowski, M.~Misiurewicz, and G.~Piliouras, ``The route to chaos in routing games: When is price of anarchy too optimistic?'' \emph{Advances in Neural Information Processing Systems}, vol.~33, pp. 766--777, 2020.

\bibitem{bielawski2021follow}
J.~Bielawski, T.~Chotibut, F.~Falniowski, G.~Kosiorowski, M.~Misiurewicz, and G.~Piliouras, ``Follow-the-regularized-leader routes to chaos in routing games,'' in \emph{International Conference on Machine Learning}.\hskip 1em plus 0.5em minus 0.4em\relax PMLR, 2021, pp. 925--935.

\bibitem{czechowski2023non}
A.~Czechowski and G.~Piliouras, ``Non-chaotic limit sets in multi-agent learning,'' \emph{Autonomous Agents and Multi-Agent Systems}, vol.~37, no.~2, p.~29, 2023.

\bibitem{blum2010routing}
A.~Blum, E.~Even-Dar, and K.~Ligett, ``Routing without regret: On convergence to nash equilibria of regret-minimizing algorithms in routing games,'' \emph{Theory of Computing}, vol.~6, no.~1, pp. 179--199, 2010.

\bibitem{blum2006routing}
------, ``Routing without regret: On convergence to nash equilibria of regret-minimizing algorithms in routing games,'' in \emph{Proceedings of the twenty-fifth annual ACM symposium on Principles of distributed computing}, 2006, pp. 45--52.

\bibitem{lam2016learning}
K.~Lam, W.~Krichene, and A.~Bayen, ``On learning how players learn: estimation of learning dynamics in the routing game,'' in \emph{2016 ACM/IEEE 7th International Conference on Cyber-Physical Systems (ICCPS)}.\hskip 1em plus 0.5em minus 0.4em\relax IEEE, 2016, pp. 1--10.

\bibitem{luo2021multiagent}
G.~Luo, H.~Zhang, H.~He, J.~Li, and F.-Y. Wang, ``Multiagent adversarial collaborative learning via mean-field theory,'' \emph{IEEE transactions on cybernetics}, vol.~51, no.~10, pp. 4994--5007, 2021.

\bibitem{huHeterogeneousBeliefsMultiPopulation2023}
\BIBentryALTinterwordspacing
S.~Hu, H.~Soh, and G.~Piliouras, ``Heterogeneous {{Beliefs}} and {{Multi-Population Learning}} in {{Network Games}},'' 2023. [Online]. Available: \url{http://arxiv.org/abs/2301.04929}
\BIBentrySTDinterwordspacing

\bibitem{li2017global}
H.~Li, H.~An, W.~Fang, Y.~Wang, W.~Zhong, and L.~Yan, ``Global energy investment structure from the energy stock market perspective based on a heterogeneous complex network model,'' \emph{Applied Energy}, vol. 194, pp. 648--657, 2017.

\bibitem{hui2019game}
Y.~Hui, Z.~Su, T.~H. Luan, and J.~Cai, ``A game theoretic scheme for optimal access control in heterogeneous vehicular networks,'' \emph{IEEE Transactions on Intelligent Transportation Systems}, vol.~20, no.~12, pp. 4590--4603, 2019.

\bibitem{2022wangModellingDynamicsRegret}
Z.~Wang, C.~Mu, S.~Hu, C.~Chu, and X.~Li, ``\BIBforeignlanguage{en}{Modelling the {{Dynamics}} of {{Regret Minimization}} in {{Large Agent Populations}}: A {{Master Equation Approach}}},'' in \emph{\BIBforeignlanguage{en}{31st {{International Joint Conference}} on {{Artificial Intelligence}}}}, 2022, p.~7.

\bibitem{kearns2001graphical}
M.~Kearns, M.~L. Littman, and S.~Singh, ``Graphical models for game theory,'' in \emph{Proceedings of the 17th Conference in Uncertainty in Artificial, Intelligence, 2001}, 2001, pp. 253--260.

\bibitem{czechowski2022poincare}
A.~Czechowski and G.~Piliouras, ``Poincar{\'e}-bendixson limit sets in multi-agent learning,'' in \emph{Proceedings of the 21st International Conference on Autonomous Agents and Multiagent Systems}, 2022, pp. 318--326.

\bibitem{sandholm2010population}
W.~H. Sandholm, \emph{Population games and evolutionary dynamics}.\hskip 1em plus 0.5em minus 0.4em\relax MIT press, 2010.

\bibitem{li2016distributed}
S.~Li, J.~He, Y.~Li, and M.~U. Rafique, ``Distributed recurrent neural networks for cooperative control of manipulators: A game-theoretic perspective,'' \emph{IEEE transactions on neural networks and learning systems}, vol.~28, no.~2, pp. 415--426, 2016.

\bibitem{liu2018modeling}
C.~Liu, E.~Zhu, Q.~Zhang, and X.~Wei, ``Modeling of agent cognition in extensive games via artificial neural networks,'' \emph{IEEE transactions on neural networks and learning systems}, vol.~29, no.~10, pp. 4857--4868, 2018.

\bibitem{goodman1953population}
L.~A. Goodman, ``Population growth of the sexes,'' \emph{Biometrics}, vol.~9, no.~2, pp. 212--225, 1953.

\bibitem{whittle1957use}
P.~Whittle, ``On the use of the normal approximation in the treatment of stochastic processes,'' \emph{Journal of the Royal Statistical Society: Series B (Methodological)}, vol.~19, no.~2, pp. 268--281, 1957.

\bibitem{1992thiemeConvergenceResultsPoincareBendixson}
H.~Thieme, ``\BIBforeignlanguage{en}{Convergence results and a {{Poincare-Bendixson}} trichotomy for asymptotically autonomous differential equations},'' \emph{\BIBforeignlanguage{en}{Journal of Mathematical Biology}}, vol.~30, no.~7, Aug. 1992.

\bibitem{2003tuylsSelectionmutationModelQlearning}
K.~Tuyls, K.~Verbeeck, and T.~Lenaerts, ``A selection-mutation model for q-learning in multi-agent systems,'' in \emph{Proceedings of the Second International Joint Conference on {{Autonomous}} Agents and Multiagent Systems}, 2003, pp. 693--700.

\bibitem{szabo2007evolutionary}
G.~Szab{\'o} and G.~Fath, ``Evolutionary games on graphs,'' \emph{Physics reports}, vol. 446, no. 4-6, pp. 97--216, 2007.

\bibitem{2019huModellingDynamicsMultiagent}
S.~Hu, C.-w. Leung, and H.-f. Leung, ``\BIBforeignlanguage{en}{Modelling the {{Dynamics}} of {{Multiagent Q-Learning}} in {{Repeated Symmetric Games}}: A {{Mean Field Theoretic Approach}}},'' in \emph{\BIBforeignlanguage{en}{33rd {{Conference}} on {{Neural Information Processing Systems}} ({{NeurIPS}} 2019)}}, 2019, p.~11.

\bibitem{fudenberg2011heterogeneous}
D.~Fudenberg and S.~Takahashi, ``Heterogeneous beliefs and local information in stochastic fictitious play,'' \emph{Games and Economic Behavior}, vol.~71, no.~1, pp. 100--120, 2011.

\bibitem{aamashu2022}
S.~Hu, C.-W. Leung, H.-f. Leung, and H.~Soh, ``The dynamics of q-learning in population games: A physics-inspired continuity equation model,'' in \emph{Proceedings of the 21st International Conference on Autonomous Agents and Multiagent Systems}, ser. AAMAS '22.\hskip 1em plus 0.5em minus 0.4em\relax Richland, SC: International Foundation for Autonomous Agents and Multiagent Systems, 2022, p. 615–623.

\bibitem{leungModellingDynamicsMultiAgent2022}
C.-w. Leung, S.~Hu, and H.-f. Leung, ``Modelling the {{Dynamics}} of {{Multi-Agent Q-learning}}: {{The Stochastic Effects}} of {{Local Interaction}} and {{Incomplete Information}},'' in \emph{Proceedings of the {{Thirty-First International Joint Conference}} on {{Artificial Intelligence}}}.\hskip 1em plus 0.5em minus 0.4em\relax {Vienna, Austria}: {International Joint Conferences on Artificial Intelligence Organization}, Jul. 2022, pp. 384--390.

\bibitem{cressman2003evolutionary}
R.~Cressman, \emph{Evolutionary dynamics and extensive form games}.\hskip 1em plus 0.5em minus 0.4em\relax MIT Press, 2003, vol.~5.

\bibitem{mckelvey1995quantal}
R.~D. McKelvey and T.~R. Palfrey, ``Quantal response equilibria for normal form games,'' \emph{Games and economic behavior}, vol.~10, no.~1, pp. 6--38, 1995.

\bibitem{markus1956asymptotically}
L.~Markus, ``Asymptotically autonomous differential systems. contributions to the theory of nonlinear oscillations iii (s. lefschetz, ed.), 17-29,'' \emph{Annals of Mathematics Studies}, vol.~36, 1956.

\bibitem{hennes2009state}
D.~Hennes, K.~Tuyls, and M.~Rauterberg, ``State-coupled replicator dynamics,'' in \emph{Proceedings of The 8th International Conference on Autonomous Agents and Multiagent Systems-Volume 2}, 2009, pp. 789--796.

\bibitem{kaisers2010frequency}
M.~Kaisers and K.~Tuyls, ``Frequency adjusted multi-agent q-learning,'' in \emph{Proceedings of the 9th International Conference on Autonomous Agents and Multiagent Systems: volume 1-Volume 1}, 2010, pp. 309--316.

\bibitem{kianercy2012dynamics}
A.~Kianercy and A.~Galstyan, ``Dynamics of boltzmann q learning in two-player two-action games,'' \emph{Physical Review E}, vol.~85, no.~4, p. 041145, 2012.

\bibitem{panozzo2014evolutionary}
F.~Panozzo, N.~Gatti, and M.~Restelli, ``Evolutionary dynamics of q-learning over the sequence form,'' in \emph{Proceedings of the AAAI Conference on Artificial Intelligence}, vol.~28, no.~1, 2014.

\bibitem{nips2021leonardos}
S.~Leonardos, G.~Piliouras, and K.~Spendlove, ``Exploration-{{Exploitation}} in {{Multi-Agent Competition}}: {{Convergence}} with {{Bounded Rationality}},'' in \emph{Advances in {{Neural Information Processing Systems}}}, vol.~34.\hskip 1em plus 0.5em minus 0.4em\relax {Curran Associates, Inc.}, 2021a, pp. 26\,318--26\,331.

\bibitem{2021leonardosExplorationExploitationMultiAgentLearning}
S.~Leonardos and G.~Piliouras, ``\BIBforeignlanguage{en}{Exploration-{{Exploitation}} in {{Multi-Agent Learning}}: {{Catastrophe Theory Meets Game Theory}}},'' \emph{\BIBforeignlanguage{en}{Proceedings of the AAAI Conference on Artificial Intelligence}}, vol.~35, no.~13, pp. 11\,263--11\,271, May 2021b.

\bibitem{cheng2014finite}
D.~Cheng, ``On finite potential games,'' \emph{Automatica}, vol.~50, no.~7, pp. 1793--1801, 2014.

\bibitem{wang2024evolutionary}
C.~Wang, M.~Perc, and A.~Szolnoki, ``Evolutionary dynamics of any multiplayer game on regular graphs,'' \emph{Nature Communications}, vol.~15, no.~1, p. 5349, 2024.

\end{thebibliography}

% \newpage
\appendices
\section{Omitted Proofs in Section 3}
\subsection{Proof of Theorem~\ref{theorem:regretdynamics}}
It follows from \Cref{eq:Ria} that the change in $\mathbf{R}^t_{i_k}$ between two discrete time steps is
\begin{equation}
    \begin{split}
        \mathbf{R}_{i_k}^{t+1}-\mathbf{R}_{i_k}^{t} & =\frac{1}{t}\sum_{\tau=1}^{t}\mathbf{r}_{i_k}^{\tau}-\frac{1}{t-1}\sum_{\tau=1}^{t-1}\mathbf{r}_{i_k}^{\tau} \\
                                                    & =\frac{1}{t}(\mathbf{r}_{i_k}^{t}-\mathbf{R}_{i_k}^{t}).
    \end{split}
\end{equation}

\begin{equation}
    \mathbf{R}_{i_k}^{t+1}=\mathbf{R}_{i_k}^{t}+\frac{1}{t}(\mathbf{r}_{i_k}^{t}-\mathbf{R}_{i_k}^{t}).
\end{equation}
Suppose that the amount of time that passes between two successive time steps is $\delta\in (0,1]$. Since $\mathbf{R}_{i_k}(t)$ is defined and smooth around time $t+\delta$, we derive the continuous-time differential equation as
\begin{equation}
    \mathbf{R}_{i_k}^{t+\delta}=
    \mathbf{R}_{i_k}^{t}+\frac{\delta}{t}(\mathbf{r}_{i_k}^{t}-\mathbf{R}_{i_k}^{t}).
\end{equation}
Note that at each step $t$, an arbitrary agent $k$ in each population $i\in V$ interacts with all opponents from adjacent large-scale populations.

Consider a population $i$. We write the change in regrets of this population as follows,
Then, we rewrite \Cref{eq:PDF_R1} as
\begin{equation}
    \mathbf{R}_{i}^{t+\delta}-\mathbf{R}_{i}^{t}=\frac{\delta}{t}(\bar{\mathbf{r}}_{i}^{t}-\mathbf{R}_{i}^{t}).
\end{equation}

Next, we consider a test function $\theta(\mathbf{R}_i)$. Define
\begin{equation}
    Y=\frac{\mathbb{E}[\theta(\mathbf{R}_i^{t+\delta})]-\mathbb{E}[\theta(\mathbf{R}_i^{t})]}{\delta}.
\end{equation}
Applying Taylor series for $\theta(\mathbf{R}_i^{t+\delta})$ at $\theta(\mathbf{R}_i^{t})$, we obtain
\begin{equation}
    \begin{split}
        \theta(\mathbf{R}_i^{t+\delta})= & \theta(\mathbf{R}_i^{t})+\frac{\delta(\bar{\mathbf{r}}_i^t-\mathbf{R}_i^t)}{t}\partial_{\mathbf{R}_i}\theta(\mathbf{R}_i)              \\
                                         & +\frac{\delta^2}{2t^2}(\bar{\mathbf{r}}_i^t-\mathbf{R}_i^t)^\top\mathbf{H}_{\theta(\mathbf{R}_i)}(\bar{\mathbf{r}}_i^t-\mathbf{R}_i^t) \\
                                         & +\mathcal{O}[\frac{\delta^2}{t^2}(\bar{\mathbf{r}}_{i}^{t}-\mathbf{R}_{i}^{t})^2],
    \end{split}
\end{equation}
where $\mathbf{H}$ denotes the Hessian matrix. Hence, the expectation $\mathbb{E}[\theta(\mathbf{R}_i^{t+\delta})]$ is
\begin{equation}
    \begin{split}
        \mathbb{E}[\theta(\mathbf{R}_i^{t+\delta})] & =\mathbb{E}[\theta(\mathbf{R}_i^{t})]+\frac{\delta}{t}\mathbb{E}[\partial_{\mathbf{R}_i}\theta(\mathbf{R}_i^t)(\bar{\mathbf{r}}_i^t-\mathbf{R}_i^t)] \\
                                                    & =\frac{\delta^2}{2t^2}\mathbb{E}[(\bar{\mathbf{r}}_i^t-\mathbf{R}_i^t)^\top\mathbf{H}_{\theta(\mathbf{R}_i^t)}(\bar{\mathbf{r}}_i^t-\mathbf{R}_i^t)] \\
                                                    & +\frac{\delta^2}{2t^2}\mathbb{E}(\mathcal{O}[(\bar{\mathbf{r}}_{i}^{t}-\mathbf{R}_{i}^{t})^2]).
    \end{split}
\end{equation}
Moving term the $\mathbb{E}[\theta(\mathbf{R}_i^{t})]$ to the left side and dividing both sides by $\delta$, we recover the quantity $Y$, i.e.,
\begin{equation}
    \begin{split}
        Y= & \frac{1}{t}\mathbb{E}[\partial_{\mathbf{R}_i}\theta(\mathbf{R}_i^t)(\bar{\mathbf{r}}_i^t-\mathbf{R}_i^t)]                                                                                                        \\
           & +\frac{\delta}{2t^2}\mathbb{E}((\bar{\mathbf{r}}_i^t-\mathbf{R}_i^t)^\top\mathbf{H}_{\theta(\mathbf{R}_i^t)}(\bar{\mathbf{r}}_i^t-\mathbf{R}_i^t)+\mathcal{O}[(\bar{\mathbf{r}}_{i}^{t}-\mathbf{R}_{i}^{t})^2]).
    \end{split}
\end{equation}
Taking the limit of $Y$ with $\delta\to 0$, the contribution of the second term on the right side vanishes, yielding
\begin{equation}
    \begin{split}
        \lim_{\delta\to 0}Y & =\frac{1}{t}\mathbb{E}[\partial_{\mathbf{R}_i}\theta(\mathbf{R}_i^t)(\bar{\mathbf{r}}_i^t-\mathbf{R}_i^t)]                              \\
                            & =\frac{1}{t}\int p(\mathbf{R}_i,t) [\partial_{\mathbf{R}_i}\theta(\mathbf{R}_i^t)(\bar{\mathbf{r}}_i^t-\mathbf{R}_i^t)]d\mathbf{R}_i^t.
    \end{split}
\end{equation}
Applying integration by parts, we have
\begin{equation}
    \lim_{\delta\to 0}Y=0-\frac{1}{t}\int \theta(\mathbf{R}_i^t)\nabla\cdot[p(\mathbf{R}_i,t)(\bar{\mathbf{r}}_i^t-\mathbf{R}_i^t)]d\mathbf{R}_i^t,
\end{equation}
where we have leveraged that the probability mass $p(\mathbf{R}_i,t)$ at the boundary $\partial\Omega_i$ remains zero since $\mathbf{R}_i^0=0$ by assumption.
Then, according to the definition of $Y$,
\begin{equation}
    \begin{split}
        \lim_{\delta\to 0}Y & =\lim_{\delta\to 0}\int \theta(\mathbf{R}_i^t)\frac{p(\mathbf{R}_i,t+\delta)-p(\mathbf{R}_i,t)}{\delta}d\mathbf{R}_i \\
                            & =\lim_{\delta\to 0}\int\theta(\mathbf{R}_i^t) \partial_t p(\mathbf{R}_i,t)d\mathbf{R}_i.
    \end{split}
\end{equation}
Hence, we have the equality
\begin{equation}
    \begin{split}
         & \int\theta(\mathbf{R}_i^t) \partial_t p(\mathbf{R}_i,t)d\mathbf{R}_i                                                        \\
         & =-\frac{1}{t}\int \theta(\mathbf{R}_i^t)\nabla\cdot[p(\mathbf{R}_i,t)(\bar{\mathbf{r}}_i^t-\mathbf{R}_i^t)]d\mathbf{R}_i^t.
    \end{split}
\end{equation}
As $\theta$ is a test function, this leads to
\begin{equation}\label{eq:dpRdt}
    \partial_t p(\mathbf{R}_i,t)=-\frac{1}{t}\nabla\cdot[p(\mathbf{R}_i,t)(\bar{\mathbf{r}}_i^t-\mathbf{R}_i^t)].
\end{equation}
Rearranging the terms, we obtain \Cref{eq:PDF_R1} in the main paper. By the definition of expectation given a probability distribution, it is straightforward to obtain \Cref{eq:Eria_int} in the main paper. This concludes the proof. Q.E.D.

\begin{corollary}
    For any population $i\in V$, the total probability mass $p(\mathbf{R}_i,t)$ always remains conserved.
\end{corollary}
\begin{proof}
    Consider the time derivative of the total probability mass
    \begin{equation}
        \frac{d}{dt}\int p(\mathbf{R}_i,t)d\mathbf{R}_i.
    \end{equation}
    Apply the Leibniz rule to interchange differentiation and integration,
    \begin{equation}
        \frac{d}{dt}\int p(\mathbf{R}_i,t)d\mathbf{R}_i=\int\frac{\partial p(\mathbf{R}_i,t)}{\partial t}d\mathbf{R}_i.
    \end{equation}
    Substitute $\frac{\partial p(\mathbf{R}_i,t)}{\partial t}$ with \Cref{eq:dpRdt},
    \begin{equation}
        \begin{split}
             & \frac{d}{dt}\int p(\mathbf{R}_i,t)d\mathbf{R}_i
            =\int\frac{\partial p(\mathbf{R}_i,t)}{\partial t}d\mathbf{R}_i                                                                \\
             & =-\frac{1}{t}\int \nabla\cdot [p(\mathbf{R}_i,t)(\bar{\mathbf{r}}_i-\mathbf{R}_i)]d\mathbf{R}_i                             \\
             & =-\frac{1}{t}\int\sum_{s\in S_i} \partial_{\bar{R}_{is}}[p(\mathbf{R}_i,t)(\bar{r}_{is}-\bar{R}_{is})]d\mathbf{R}_i         \\
             & =-\frac{1}{t}\int\sum_{s\in S_i} [\partial_{\bar{R}_{is}}p(\mathbf{R}_i,t)](\bar{r}_{is}-\bar{R}_{is})d\mathbf{R}_i         \\
             & \;\;\;\;-\frac{1}{t}\int p(\mathbf{R}_i,t)\sum_{s\in S_i}[\partial_{\bar{R}_{is}}(\bar{r}_{is}-\bar{R}_{is})]d\mathbf{R}_i.
        \end{split}
    \end{equation}
    Apply integration by parts,
    \begin{equation}
        \begin{split}
             & \int\sum_{s\in S_i}[\partial_{R_{is}}p(\mathbf{R}_i,t)](\bar{r}_{is}-R_{is})d\mathbf{R}_i    \\
             & =0-\int p(\mathbf{R}_i,t)\sum_{s\in S_i}\partial_{R_{is}}(\bar{r}_{is}-R_{is})d\mathbf{R}_i,
        \end{split}
    \end{equation}
    where we have leveraged that the probability mass $p(\mathbf{R}_i,t)$ at the boundary remains zero due to the boundary condition $\mathbf{R}_i^0=0$. Hence, we have
    \begin{equation}
        \frac{d}{dt}\int p(\mathbf{R}_i,t)d\mathbf{R}_i=0.
    \end{equation}
\end{proof}

\subsection{Proof of Theorem~\ref{theorem:heto_dERdt}}
The time derivative of the mean regret $\bar{R}_{ia}$ of not having chosen action $a\in S_i$ is
\begin{equation}
    \frac{d\mathbb{E}(R_{ia})}{dt}=\frac{d}{dt}\int R_{ia}p(\mathbf{R}_i,t) d\mathbf{R}_i.
\end{equation}
Applying the Leibniz rule to interchange differentiation and integration, we have
\begin{equation}
    \frac{d\mathbb{E}(R_{ia})}{dt}=\frac{d}{dt}\int R_{ia}p(\mathbf{R}_i,t) d\mathbf{R}_i=\int R_{ia} \frac{\partial p(\mathbf{R}_i,t)}{\partial t} d\mathbf{R}_i.
\end{equation}
Then substitute $\frac{\partial p(\mathbf{R}_i,t)}{\partial t}$ with \Cref{eq:dpRdt}
\begin{equation}
    \begin{split}
        \frac{d\mathbb{E}(R_{ia})}{dt} & =-\frac{1}{t}\int R_{ia}\nabla\cdot [p(\mathbf{R}_i,t)(\bar{\mathbf{r}}_i-\mathbf{R}_i)]d\mathbf{R}_i                  \\
                                       & =-\frac{1}{t}\int R_{ia} \sum_{s\in S_i} \partial_{R_{is}}[p(\mathbf{R}_i,t)(\bar{r}_{is}-R_{is})] d\mathbf{R}_i       \\
                                       & =-\frac{1}{t} \int R_{ia} \sum_{s\in S_i} (\partial_{R_{is}}[p(\mathbf{R}_i,t)])(\bar{r}_{is}-R_{is}) d\mathbf{R}_i    \\
                                       & \;\;\;\;-\frac{1}{t}\int R_{ia}p(\mathbf{R}_i,t)\sum_{s\in S_i}[\partial_{R_{is}}(\bar{r}_{is}-R_{is})] d\mathbf{R}_i. \\
    \end{split}
\end{equation}
Apply integration by parts to the first term
\begin{equation}\label{eq:int1}
    \begin{split}
         & -\frac{1}{t} \int R_{ia} \sum_{s\in S_i} [\partial_{\bar{R}_{is}}p(\mathbf{R}_i,t)](\bar{r}_{is}-R_{is}) d\mathbf{R}_i   \\
         & =-\frac{1}{t}(0-\int p(\mathbf{R}_i,t)\partial_{R_{ia}}[R_{ia}(\bar{r}_{ia}-R_{ia})]d\mathbf{R}_i                        \\
         & \;\;\;\;-\int p(\mathbf{R}_i,t)R_{ia}\sum_{a'\neq a}[\partial_{R_{ia'}}(\bar{r}_{ia'}-R_{ia'})]d\mathbf{R}_i)            \\
         & =\frac{1}{t}\int p(\mathbf{R}_i,t)\partial_{R_{ia}}[R_{ia}(\bar{r}_{ia}-R_{ia})]d\mathbf{R}_i                            \\
         & \;\;\;\;+\frac{1}{t}\int p(\mathbf{R}_i,t)R_{ia}\sum_{a'\neq a}[\partial_{R_{ia'}}(\bar{r}_{ia'}-R_{ia'})]d\mathbf{R}_i,
    \end{split}
\end{equation}
where we have leveraged that the probability mass at the boundary remains zero. Hence, following \Cref{eq:int1}, we have
\begin{equation}\label{eq:dEriadt}
    \begin{split}
        \frac{d\mathbb{E}(R_{ia})}{dt} & =\frac{1}{t}\int p(\mathbf{R}_i,t)\partial_{R_{ia}}[R_{ia}(\bar{r}_{ia}-R_{ia})]d\mathbf{R}_i                           \\
                                       & \;\;\;\;+\frac{1}{t}\int p(\mathbf{R}_i,t)R_{ia}\sum_{a'\neq a}[\partial_{R_{ia'}}(\bar{r}_{ia'}-R_{ia'})]d\mathbf{R}_i \\
                                       & \;\;\;\;-\frac{1}{t}\int R_{ia}p(\mathbf{R}_i,t)\sum_{s\in S_i}[\partial_{R_{is}}(\bar{r}_{is}-R_{is})] d\mathbf{R}_i   \\
                                       & =\frac{1}{t}\int p(\mathbf{R}_i,t)(\bar{r}_{ia}-R_{ia})d\mathbf{R}_i                                                    \\
                                       & \;\;\;\;+\frac{1}{t}\int p(\mathbf{R}_i,t)R_{ia}\partial_{R_{ia}}(\bar{r}_{is}-R_{is})d\mathbf{R}_i                     \\
                                       & \;\;\;\;-\frac{1}{t}\int R_{ia}p(\mathbf{R}_i,t)\partial_{R_{ia}}(\bar{r}_{is}-R_{is})d\mathbf{R}_i                     \\
                                       & =\frac{1}{t}\bar{r}_{ia}\int p(\mathbf{R}_i,t)d\mathbf{R}_i-\frac{1}{t}\int R_{ia}p(\mathbf{R}_i,t)d\mathbf{R}_i        \\
                                       & =\frac{\bar{r}_{ia}-\bar{R}_{ia}}{t}.
    \end{split}
\end{equation}

Given that agents hold separate regrets, the joint probability density function $p(\mathbf{R},t)$ can be factorized as the product of each neighbor population regrets, i.e., $p(\mathbf{R},t)=\prod_j\in N p(\mathbf{R}_j,t)$. Then we repeat the mean instantaneous regret $\bar{r}_{ia}$, which has been given in \Cref{eq:Eria_int} that
\begin{equation}\label{eq:Eria_int2}
    \begin{split}
        \bar{r}_{ia}= & \int\int_{\prod_{j\in V_i}}[\mathbf{e}_{ia}-\mathbf{x}_i]^\top\sum_{j\in V_i}\frac{1}{\vert V_j\vert}\mathbf{A}_{ij}\bar{\mathbf{x}}_j \\
                      & (p(\mathbf{R}_i, t)\prod_{j\in V_i}p(\mathbf{R}_j,t))(d\mathbf{R}_i\prod_{j\in V_i}d\mathbf{R}_j),
    \end{split}
\end{equation}
where $\bar{\mathbf{x}}_j=\int_{\prod_{b\in S_j}}\frac{\exp(\lambda\mathbf{R}_{jb})}{\sum_{b'\in S_j}\exp(\lambda\mathbf{R}_{jb'})}(\prod_{b\in S_j}d\mathbf{R}_{jb})$.
Define $V_{ij}\coloneqq \{i\}+V_i$, $\bar{\mathbf{R}}\coloneqq\{\bar{\mathbf{R}}_i, \bar{\mathbf{R}}_j\}_{j\in V_i}$ and
\begin{equation}\label{eq:sfia}
    f_{ia}(\{\mathbf{R}_j\}_{j\in V_i})\coloneqq f_{ia}(\mathbf{R}) = [\mathbf{e}_{ia}-\mathbf{x}_i]^\top\sum_{j\in V_i}\frac{1}{\vert V_j\vert}\mathbf{A}_{ij}\bar{\mathbf{x}}_j.
\end{equation}
Applying Taylor series for $f_{ia}(\mathbf{R})$ at the mean regret $f_{ia}(\bar{\mathbf{R}})$, we have
\begin{equation}
    \begin{split}
        f_{ia}(\mathbf{R}) & \approx f_{ia}(\bar{\mathbf{R}})+\nabla f_{ia}(\bar{\mathbf{R}})(\mathbf{R}-\bar{\mathbf{R}})                                                                            \\
                           & +\frac{1}{2}(\mathbf{R}-\bar{\mathbf{R}})^\top \mathbf{H}_{f_{ia}(\bar{\mathbf{R}})}(\mathbf{R}-\bar{\mathbf{R}})+\mathcal{O}(\Vert \mathbf{R}-\bar{\mathbf{R}}\Vert^3),
    \end{split}
\end{equation}
where $\mathbf{H}$ denotes the Hessian matrix. Hence, we can rewrite \Cref{eq:Eria_int2}
\begin{scriptsize}
    \begin{equation}
        \begin{split}
            \bar{r}_{ia} & =\int \int_{\prod_{j\in V_i}} f_{ia}(\{\mathbf{R}_j\}_{j\in V_i})(p(\mathbf{R}_i,t)\prod_{j\in V_i}p(\mathbf{R}_j,t))(d\mathbf{R}_i\prod_{j\in V_i}d\mathbf{R}_j)                                      \\
                         & \approx f_{ia}(\bar{\mathbf{R}})+\int\nabla f_{ia}(\bar{\mathbf{R}})\mathbf{R}(\prod_{h\in V_{ij}}p(\mathbf{R}_h,t))(\prod_{h\in V_{ij}}d\mathbf{R}_h)-\nabla f_{ia}(\bar{\mathbf{R}})\bar{\mathbf{R}} \\
                         & \;\;\;\;+\frac{1}{2}\int(\mathbf{R}-\bar{\mathbf{R}})^\top \mathbf{H}_{f_{ia}(\bar{\mathbf{R}})}(\mathbf{R}-\bar{\mathbf{R}})(\prod_{h\in V_{ij}}p(\mathbf{R}_h,t))(\prod_{h\in V_{ij}}d\mathbf{R}_h)  \\
                         & \;\;\;\;+\int\mathcal{O}(\Vert \mathbf{R}-\bar{\mathbf{R}}\Vert^3)(\prod_{h\in V_{ij}}p(\mathbf{R}_h,t))(\prod_{h\in V_{ij}}d\mathbf{R}_h).
        \end{split}
    \end{equation}
\end{scriptsize}

It's easy to observe that the second and the third term can be offset. Moreover, for any two populations $h, h'\in V_{ij}$, the regrets $\mathbf{R}_h$, $\mathbf{R}_{h'}$ about these two populations are separate and independent. Hence, the covariance of these regrets is zero. By applying the moment closure approximation, we have
\begin{equation}\label{eq:eria}
    \begin{split}
        \bar{r}_{ia} & \approx f_{ia}(\bar{\mathbf{R}})+\frac{1}{2}\sum_{h\in V_{ij}}\sum_{s_h\in S_h}[\frac{\partial f_{ia}(\bar{\mathbf{R}})}{\partial R_{hs_h}}]^2\operatorname{Var}(R_{hs_h}) \\
                     & =f_{ia}(\bar{\mathbf{R}})+\frac{1}{2}[f''(\{\bar{\mathbf{R}}_j\}_{j\in V_i})\operatorname{Var}(\{\mathbf{R}_j\}_{j\in V_i})]
    \end{split}
\end{equation}
Substitute \Cref{eq:eria} back to \Cref{eq:dEriadt},
\begin{small}
    \begin{equation}
        \frac{d\bar{R}_{ia}}{dt}\approx \frac{1}{t}[f_{ia}(\{\bar{\mathbf{R}}_j\}_{j\in V_i})-\bar{R}_{ia}]+\frac{1}{2t}[f''(\{\bar{\mathbf{R}}_j\}_{j\in V_i})\operatorname{Var}(\{\mathbf{R}_j\}_{j\in V_i})].
    \end{equation}
\end{small}
This concludes the proof.

\subsection{Proof of Theorem~\ref{theorem:VarDecay}}
Without loss of generality, we consider the variance of the regret $\operatorname{Var}(R_{ia})$ about action $a$ of population $i$. Note that
\begin{equation}
    \operatorname{Var}(R_{ia})=\mathbb{E}[R_{ia}^2]-(\bar{R}_{ia})^2.
\end{equation}
Hence, we have
\begin{equation}\label{eq:svardR}
    \begin{split}
        \frac{d\operatorname{Var}(R_{ia})}{dt}= \frac{d\mathbb{E}[R_{ia}^2]}{dt}-2\bar{R}_{ia}\frac{d \bar{R}_{ia}}{dt}.
    \end{split}
\end{equation}
Applying the Leibniz rule to interchange differentiation and integration, we have
\begin{equation}
    \frac{d\mathbb{E}(R_{ia}^2)}{dt}=\frac{d}{dt}\int R_{ia}^2p(\mathbf{R}_i,t) d\mathbf{R}_i=\int R_{ia}^2\frac{\partial p(\mathbf{R}_i,t)}{\partial t} d\mathbf{R}_i.
\end{equation}
Then substitute $\frac{\partial p(\mathbf{R}_i,t)}{\partial t}$ with \Cref{eq:dpRdt}
\begin{equation}\label{eq:dria2dt}
    \begin{split}
        \frac{d\mathbb{E}(R_{ia}^2)}{dt} & =-\frac{1}{t}\int R_{ia}^2\nabla\cdot[p(\mathbf{R}_i,t)(\bar{\mathbf{r}}_i-\mathbf{R}_i)]d\mathbf{R}_i                       \\
                                         & =-\frac{1}{t}\int R_{ia}^2\sum_{s\in S_i}\partial_{R_{is}}[p(\mathbf{R}_i,t)(\bar{r}_{is}-R_{is})]d\mathbf{R}_i              \\
                                         & =-\frac{1}{t}\int R_{ia}^2\sum_{s\in S_i}[\partial_{R_{is}}p(\mathbf{R}_i,t)](\bar{r}_{is}-R_{is})d\mathbf{R}_i              \\
                                         & \;\;\;\;-\frac{1}{t}\int R_{ia}^2 p(\mathbf{R}_i,t)\sum_{s\in S_i}[\partial_{R_{is}}(\bar{r}_{is}-R_{is}]d\mathbf{R}_i       \\
                                         & =\frac{1}{t}\int p(\mathbf{R}_i,t) \sum_{s\in S_i}\partial_{R_{is}}[R_{ia}^2
        (\bar{r}_{is}-R_{is})]d\mathbf{R}_i                                                                                                                             \\
                                         & \;\;\;\;-\frac{1}{t}\int R_{ia}^2 p(\mathbf{R}_i,t)\sum_{s\in S_i}[\partial_{R_{is}}(\bar{r}_{is}-R_{is})]d\mathbf{R}_i      \\
                                         & =\frac{1}{t}\int p(\mathbf{R}_i,t) \partial_{R_{ia}}[R_{ia}^2
        (\bar{r}_{ia}-R_{ia})]d\mathbf{R}_i                                                                                                                             \\
                                         & \;\;\;\;+\frac{1}{t}\int p(\mathbf{R}_i,t)R_{ia}^2 \sum_{a'\neq a}\partial_{R_{ia'}}[(\bar{r}_{ia'}-R_{ia'})]d\mathbf{R}_i   \\
                                         & \;\;\;\;-\frac{1}{t}\int R_{ia}^2 p(\mathbf{R}_i,t)\sum_{s\in S_i}[\partial_{R_{is}}(\bar{r}_{is}-R_{is})]d\mathbf{R}_i      \\
                                         & =\frac{1}{t}\int 2p(\mathbf{R}_i,t) R_{ia}(\bar{r}_{ia}-R_{ia})d\mathbf{R}_i                                                 \\
                                         & \;\;\;\;+\frac{1}{t}\int p(\mathbf{R}_i,t) R_{ia}^2 \partial_{R_{ia}}(\bar{r}_{ia}-R_{ia})d\mathbf{R}_i                      \\
                                         & \;\;\;\;-\frac{1}{t}\int p(\mathbf{R}_i,t) R_{ia}^2\partial_{R_{ia}}(\bar{r}_{ia}-R_{ia})d\mathbf{R}_i                       \\
                                         & =\frac{1}{t}\int 2p(\mathbf{R}_i,t) R_{ia}\bar{r}_{ia}d\mathbf{R}_i-\frac{1}{t}\int 2p(\mathbf{R}_i,t) R_{ia}^2d\mathbf{R}_i \\
                                         & =\frac{2\bar{r}_{ia}\bar{R}_{ia}-2\mathbb{E}[R_{ia}^2]}{t}.
    \end{split}
\end{equation}
Combining \Cref{eq:dEriadt} and \Cref{eq:dria2dt}, we have
\begin{equation}
    \begin{split}
        \frac{d\operatorname{Var}(R_{ia})}{dt} & = \frac{d\mathbb{E}[R_{ia}^2]}{dt}-2\bar{R}_{ia}\frac{d\bar{R}_{ia}}{dt}                                 \\
                                               & =\frac{2\bar{r}_{ia}\bar{R}_{ia}-2\mathbb{E}[R_{ia}^2]}{t}-\frac{2\bar{r}_{ia}\bar{R}_{ia}-2R_{ia}^2}{t} \\
                                               & =-\frac{2\mathbb{E}[R_{ia}^2]-2R_{ia}^2}{t}                                                              \\
                                               & =-\frac{2\operatorname{Var}(R_{ia})}{t}.
    \end{split}
\end{equation}
Rearrange the equation and apply integration; then, we have
\begin{equation}
    \begin{split}
        \frac{d\operatorname{Var}(R_{ia})}{dt}                              & =-\frac{2\operatorname{Var}(R_{ia})}{t} \\
        \int\frac{1}{\operatorname{Var}(R_{ia})}d\operatorname{Var}(R_{ia}) & =-\int\frac{2}{t}dt                     \\
        \operatorname{Var}(R_{ia})                                          & =\frac{1}{t^2}\sigma^2(R_{ia}),
    \end{split}
\end{equation}
where $\sigma^2(R_{ia})$ is the initial variance of $R_{ia}$. This concludes the proof.

\section{Omitted Proof in Section 4}
\subsection*{Proof of Theorem~\ref{theorem:convergence}}

In this section, we develop the necessary technical framework for the proof of Theorem~\ref{theorem:convergence}.
Theorem~\ref{theorem:convergence} is based upon \textcolor{black}{four} critical lemmas.

Prior to starting our convergence proof, we leverage the notion of the asymptotically autonomous dynamical system formally defined in \ref{def:asy}.

First of all, we prove Lemma~\ref{lemma:heto_dERdt}. Making use of the chain rule, we obtain
\begin{equation}\label{eq:dRidtau}
    \frac{d\mathbb{E}[\mathbf{R}_i]}{d\tau}=\frac{d\mathbb{E}[\mathbf{R}_i]}{dt}\frac{dt}{d\tau}=\frac{\bar{\mathbf{r}}_i-\mathbf{R}_i}{e^\tau}e^\tau=\bar{\mathbf{r}}_i-\mathbf{R}_i.
\end{equation}
Theorem~\ref{theorem:VarDecay} states that the variance of beliefs
eventually goes to zero, implying that over time the system heterogeneity will vanish and agents will develop the same regrets and play the same policy.
Then the mean regret of such a homogeneous population is equivalent to the individual agent's regret among the population.
We can write the dynamic for the homogeneous population where agents hold the same policy after time reparameterization as
\begin{equation}
    \frac{d\mathbf{R}_i}{d\tau}=\mathbf{r}_i-\mathbf{R}_i, \forall i\in V.
\end{equation}
In conjunction with Definition~\ref{def:asy}, we establish Lemma~\ref{lemma:heto_dERdt}.

Then, we prove Lemma~\ref{lemma:dxdt}.
Applying the chain rule, we have
\begin{equation}
    \begin{split}
        \frac{\partial x_{ia}}{\partial\tau} & =\frac{\partial x_{ia}}{\partial R_{ia}}\frac{\partial R_{ia}}{\partial \tau}+\cdots+\frac{\partial x_{ia}}{\partial R_{is}}\frac{\partial R_{is}}{\partial \tau}           \\
                                             & =\frac{\partial x_{ia}}{\partial R_{ia}}\frac{\partial R_{ia}}{\partial \tau}+\sum_{a'\neq a}\frac{\partial x_{ia}}{\partial R_{ia'}}\frac{\partial R_{ia'}}{\partial \tau}
    \end{split}
\end{equation}
Consider the term $\frac{\partial x_{ia}}{\partial R_{ia}}$ based on \Cref{eq:xia},
\begin{equation}
    \begin{split}
         & \ln x_{ia}=\ln\frac{\exp(\lambda R_{ia})}{\sum_{s\in S_i}\exp(\lambda R_{is})}=\lambda R_{ia}-\ln\sum_{s\in S_i}\exp(\lambda R_{is}),         \\
         & \frac{\partial\ln x_{ia}}{\partial R_{ia}}=\lambda-\frac{\lambda\exp(\lambda R_{ia})}{\sum_{s\in S_i}\exp(\lambda R_{is})}=\lambda(1-x_{ia}).
    \end{split}
\end{equation}
Apply the chain rule,
\begin{equation}
    \frac{\partial x_{ia}}{\partial R_{ia}}=\lambda x_{ia}(1-x_{ia}).
\end{equation}
Then, apply the same process to $\frac{\partial x_{ia}}{\partial R_{ia'}}$
\begin{equation}
    \frac{\partial x_{ia}}{\partial R_{ia'}}=\frac{\partial x_{ia}}{\partial \ln x_{ia}}\frac{\partial \ln x_{ia}}{\partial R_{ia'}}=-\lambda x_{ia}x_{ia'},
\end{equation}
where $a'\neq a \in S_i$.
Hence, we have
\begin{equation}
    \frac{\partial x_{ia}}{\partial t}=x_{ia}(1-x_{ia})\frac{\partial R_{ia}}{\partial t}-\sum_{a'\neq a}\lambda x_{ia}x_{ia'}\frac{\partial R_{ia'}}{\partial t}
\end{equation}
Substitute $\frac{\partial R_{ia}}{\partial\tau}$ to above equation,
\begin{equation}
    \begin{split}
        \frac{\partial x_{ia}}{\partial \tau} & =\lambda x_{ia}(1-x_{ia})(r_{ia}-R_{ia})-\sum_{a'\neq a}\lambda x_{ia}x_{ia'}(r_{ia'}-R_{ia'}) \\
                                              & =\lambda x_{ia}[(r_{ia}-R_{ia})-\sum_{s\in S_i}x_{is}(r_{is}-R_{is})].
    \end{split}
\end{equation}
Since $\sum_{s\in S_i}x_{is}=1$, we have
\begin{equation}\label{eq:dxiadt2}
    \frac{\partial x_{ia}}{\partial \tau}=\lambda x_{ia}[r_{ia}-\sum_{s\in S_i}x_{is}r_{is}-\sum_{s\in S_i}x_{is}(R_{ia}-R_{is})].
\end{equation}
Since $\frac{x_{ia}}{x_{ia'}}=\frac{\exp (\lambda R_{ia})}{\exp (\lambda R_{ia'})}$, we obtain
\begin{equation}\label{eq:exptrick}
    \sum_{s\in S_i}x_{is}\ln(\frac{x_{is}}{x_{ia}})=\lambda \sum_{s\in S_i}x_{is}(R_{is}-R_{ia}),
\end{equation}
which give us
\begin{equation}
    \frac{\partial x_{ia}}{\partial \tau}=x_{ia}[\lambda (r_{ia}-\sum_{s\in S_i}x_{is}r_{is})+\sum_{s\in S_i}x_{is}\ln(\frac{x_{is}}{x_{ia}})].
\end{equation}
Recalling the definition of $r_{ia}$, we arrange the $r_{ia}-\sum_{s\in S_i}x_{is}r_{is}$ term as
\begin{small}
    \begin{equation}
        \begin{split}
             & r_{ia}-\sum_{s\in S_i}x_{is}r_{is}                                                                                                                                        \\
             & =u_{i}(a,\bar{\mathbf{x}}_{-i})-u_{i}(\mathbf{x}_i,\bar{\mathbf{x}}_{-i})-\sum_{s\in S_i}x_{is}[u_{i}(s,\bar{\mathbf{x}}_{-i})-u_{i}(\mathbf{x}_i,\bar{\mathbf{x}}_{-i})] \\
             & =u_{i}(a,\bar{\mathbf{x}}_{-i})-\sum_{s\in S_i}x_{is}u_{i}(s,\bar{\mathbf{x}}_{-i})                                                                                       \\
             & =e_{ia}^\top\mathbf{u}_i(\bar{\mathbf{x}}_{-i})-\mathbf{x}_i\mathbf{u}_i(\bar{\mathbf{x}}_{-i}).
        \end{split}
    \end{equation}
\end{small}
where $\mathbf{u}_i(\bar{\mathbf{x}}_{-i})=\sum_{j\in V_i}\frac{1}{\vert V_j\vert}\mathbf{A}_{ij}\bar{\mathbf{x}}_{j}$.
Substitute \Cref{eq:exptrick} back to \Cref{eq:dxiadt2},
\begin{equation}
    \frac{\partial x_{ia}}{\partial \tau}=x_{ia}[\lambda u_i(a,\bar{\mathbf{x}}_{-i})-\lambda\mathbf{x}_i\mathbf{u}_i(\bar{\mathbf{x}}_{-i})+\sum_{s\in S_i}x_{is}\ln(\frac{x_{is}}{x_{ia}})].
\end{equation}
Rearranging the terms, we obtain \Cref{eq:dRidtau}. In conjunction with Definition~\ref{def:WZSPNG}, we obtain Lemma~\ref{lemma:dxdt}.

Now we prove Theorem~\ref{theorem:convergence}.
We observe that \Cref{eq:xlimiteq} (the limit equation of policy dynamics) in Lemma~\ref{lemma:dxdt} is reduced to a smooth Q-learning dynamic with the learning rate $\alpha=1$.
Leonardos et al.~\cite{nips2021leonardos, 2021leonardosExplorationExploitationMultiAgentLearning} demonstrate the convergence in n-player weighted zero-sum network games and n-player weighted potential games.
We repeat their convergence result as follows.

To figure out whether the large-time behavior of our policy dynamics is in perfect alignment with that of the SQL dynamic,
we leverage  Thieme~\cite{1992thiemeConvergenceResultsPoincareBendixson}'s seminal result on the connection between the limit behavior of an asymptotically autonomous system $\Phi$ (i.e. the policy dynamics) and its limit equation $\Theta$ (i.e. \Cref{eq:xlimiteq}).

This means that if the equilibria of $\Theta$ are isolated and any solution of $\Theta$ converges to one of them, every solution of $\Phi$ converges to an equilibrium of $\Theta$ as well. Combining the above lemmas, every solution of our policy dynamics also converges to a unique QRE in a weighted zero-sum game and to a compact connected set of QRE in a weighted potential game. Q.E.D.

\end{document}